\documentclass[oneside,english]{amsart}
\usepackage[T1]{fontenc}
\usepackage[latin9]{inputenc}
\usepackage{geometry}
\usepackage{units}
\usepackage{mathrsfs}
\usepackage{amsthm}
\usepackage{amstext}
\usepackage{amssymb}
\usepackage{graphicx}
\usepackage{esint}

\makeatletter
\numberwithin{equation}{section}
\numberwithin{figure}{section}
\theoremstyle{plain}
\newtheorem{thm}{\protect\theoremname}
  \theoremstyle{definition}
  \newtheorem{defn}[thm]{\protect\definitionname}
  \theoremstyle{remark}
  \newtheorem{rem}[thm]{\protect\remarkname}
  \theoremstyle{plain}
  \newtheorem{prop}[thm]{\protect\propositionname}
  \theoremstyle{plain}
  \newtheorem{lem}[thm]{\protect\lemmaname}

\usepackage{amsthm}
\usepackage[mathscr]{euscript}

\newtheorem{assumption}{Assumption}

\theoremstyle{definition}

\newcommand{\V}{\mathcal{V}}
\newcommand{\E}{\mathcal{E}}

\def\N{\mathbb{N}}

\newcommand{\ui}{\textrm{i}}
\newcommand{\ue}{\textrm{e}}

\newcommand{\ud}{\mathrm{d}}

\newcommand{\bdm}{\begin{displaymath}}
\newcommand{\edm}{\end{displaymath}}
\newcommand{\beq}{\begin{equation}}
\newcommand{\eeq}{\end{equation}}
\newcommand{\beqa}{\begin{eqnarray}}
\newcommand{\eeqa}{\end{eqnarray}}

\newcommand{\id}{\mathbb{I}}

\makeatother

\usepackage{babel}
  \providecommand{\definitionname}{Definition}
  \providecommand{\lemmaname}{Lemma}
  \providecommand{\propositionname}{Proposition}
  \providecommand{\remarkname}{Remark}
\providecommand{\theoremname}{Theorem}

\begin{document}
\global\long\def\disgraph{\mathcal{G}}

\global\long\def\metgraph{\Gamma}

\global\long\def\dismetgraph{\mathcal{G}_{\Gamma}}

\global\long\def\V{\mathcal{V}}

\global\long\def\E{\overleftrightarrow{\mathcal{E}}}

\global\long\def\EE{\mathcal{E}}

\global\long\def\R{\mathbb{R}}

\global\long\def\Q{\mathbb{Q}}

\global\long\def\Z{\mathbb{Z}}

\global\long\def\C{\mathbb{C}}

\global\long\def\N{\mathbb{N}}

\global\long\def\lap{\mathbf{L}}

\global\long\def\hamil{\mathcal{H}}

\global\long\def\trans{\mathbf{P}}

\global\long\def\conmat{\mathbf{C}}

\global\long\def\hes{\mathscr{H}}

\global\long\def\morse{\mathbf{\mathscr{M}}}

\global\long\def\nums{\mathcal{N}}

\global\long\def\at{\mathbf{\big|}}

\global\long\def\id{\mathbf{1}}

\global\long\def\ui{\mathbf{\textrm{i}}}

\global\long\def\ue{\mathbf{\textrm{e}}}

\global\long\def\ud{\mathbf{\textrm{d}}}

\global\long\def\fc{\phi_{n}}

\global\long\def\sur{\sigma_{n}}

\global\long\def\nc{\nu_{n}}

\global\long\def\magv{\vec{\alpha}}

\global\long\def\lenv{\vec{l}}

\global\long\def\intlen{j}

\global\long\def\intlenv{\vec{j}}

\global\long\def\lenmat{\mathbf{E}}

\global\long\def\lenmap{\mathscr{L}}

\global\long\def\scatmat{\mathbf{S}}

\global\long\def\magp{\mathcal{A}}

\global\long\def\magmat{\mathbf{A}}

\global\long\def\degmat{\mathbf{D}}

\global\long\def\torus{\mathbb{T}^{\left|I\right|}}

\global\long\def\tvec{\vec{x}}

\global\long\def\zvec{\vec{0}}

\global\long\def\zmat{\mathbf{0}}

\global\long\def\ham{\mathcal{H}}

\global\long\def\pc{\mathcal{\mathcal{PC}}}

\global\long\def\bra#1{\left\langle #1\right|}

\global\long\def\ket#1{\left|#1\right\rangle }

\title[$\left\{ 0,1,2,3,\ldots\right\} $ implies tree]{The Nodal Count $\left\{ 0,\,1,\,2,\,3,\ldots\right\} $\\
Implies the Graph is a Tree}

\author{Ram Band}

\address{Department of Mathematics, University Walk, Clifton, Bristol BS8
1TW, U.K.}
\begin{abstract}
Sturm's oscillation theorem states that the $n^{\textrm{th}}$ eigenfunction
of a Sturm-Liouville operator on the interval has $n-1$ zeros (nodes)
\cite{Stu_jmpa36,Stu_jmpa36a}. This result was generalized for all
metric tree graphs \cite{PokPryObe_mz96,Schapotschnikow06} and an
analogous theorem was proven for discrete tree graphs \cite{Berkolaiko07,DhaRam_prl85,Fiedler_cmj75}.
We prove the converse theorems for both discrete and metric graphs.
Namely, if for all $n$, the $n^{\textrm{th}}$ eigenfunction of the
graph has $n-1$ zeros then the graph is a tree. Our proofs use a
recently obtained connection between the graph's nodal count and the
magnetic stability of its eigenvalues \cite{Ber_arx11,BerWey_ptrs13,cdv_arx12}.
In the course of the proof we show that it is not possible for all
(or even almost all, in the metric case) the eigenvalues to exhibit
a diamagnetic behaviour. In addition, we develop a notion of 'discretized'
versions of a metric graph and prove that their nodal counts are related
to this of the metric graph.
\end{abstract}
\maketitle

\section{Introduction\label{sec:Introduction}}

Nodal domains were first presented in full glory by Chladni's Sound
figures. By the end of the $18^{\textrm{th}}$ century Chladni was
performing the following demonstration: he spread sand on a brass
plate and stroke it with a violin bow. This caused the sand to accumulate
in intricate patterns of nodal lines - the lines where the vibration
amplitude vanishes. The areas bounded by the nodal lines are the nodal
domains. The first rigorous result on nodal domains is probably Sturm's
oscillation theorem, according to which a vibrating string is divided
into exactly $\textit{n}$ nodal intervals by the zeros of its $\textit{n}^{th}$
vibrational mode \cite{Stu_jmpa36,Stu_jmpa36a} (and see also \cite{CourantHilbert_volume1},
p. 454). In the next century Courant treated vibrating membranes and
proved that the number of nodal domains of the $n^{\textrm{th}}$
eigenfunction of the Laplacian is bounded from above by $n$ \cite{Courant23}
(see also \cite{CourantHilbert_volume1}, p. 452). Pleijel further
restricted the possible nodal domain counts and showed for example,
that the Courant bound can be attained only a finite number of times
\cite{Pleijel56}. These are some of the earlier results in the field
of counting nodal domains. This field had gained an exciting turn
when Blum, Gnutzmann and Smilansky have shown that the nodal count
statistics may reveal the nature of the underlying manifold - whether
its classical dynamics are integrable or chaotic, \cite{BGS02}. This
opened a new research direction of treating the nodal count from the
inverse problems perspective. One aspect of this research is rephrasing
the famous question of Mark Kac by asking 'Can one \emph{count} the
shape of a drum?' (see \cite{Kac66} for Kac's original question).
Namely, what can one learn about an object (manifold, graph, etc.),
knowing the nodal counts of all of its eigenfunctions? One way to
treat this question is by studying the direct (rather than inverse)
problem and developing formulae which describe nodal count sequences
(\cite{AroBanFajGnu_jpa12,AroSmi_arx10,BanBerSmi_ahp12,GKS06,KS08,Klawonn09}).
Such formulae are expected to reveal various geometric properties
of the underlying object which one seeks to reveal. One can also study
inverse nodal problems by comparing the nodal information with the
spectral one. A well established conjecture in the field claims that
isospectral objects have different nodal count sequences. After its
first appearance in a paper by Gnutzmann, Smilansky and Sondergaard
\cite{GSS05}, the conjecture initiated a series of works on nodal
counts of various objects, either affirming the conjecture in certain
settings \cite{BOS08,BSS06,BS07,BruKlaPuh_jpa07,Oren07}, or pointing
out counterexamples \cite{BruFaj_cmp12,OreBan_jpa12}. The general
validity of this conjecture is still not well understood. The most
recent approach in the study of nodal counts describes the nodal domain
count (and even morphology) of individual eigenfunctions in terms
of partitions. This was triggered by Helffer, Hoffmann\textendash{}Ostenhof
and Terracini who study Schr�dinger operators on two-dimensional domains
using partitions of the domain, \cite{HHT09}. They provide characterization
of the morphology and number of nodal domains of eigenfunctions which
attain the Courant bound. Following this, Band, Berkolaiko, Raz and
Smilansky used a partition approach for metric graphs and described
the number of nodal points and their location for all eigenfunctions
via a Morse index of an energy function, \cite{BanBerRazSmi_cmp12}.
This result initiated similar connections between nodal counts and
Morse indices, for discrete graphs and manifolds, which were studied
by Berkolaiko, Kuchment, Raz and Smilasnky \cite{BerKucSmi_arx11,BerRazSmi_jpa12}.
Finally, three new works show that the response of an eigenvalue to
magnetic fields on the graph dictates the nodal count via a Morse
index connection. The first of these results was obtained by Berkolaiko
for discrete graphs \cite{Ber_arx11} and shortly afterward an additional
proof was supplied by Colin de Verdi\`{e}re \cite{cdv_arx12}. Berkolaiko
and Weyand provide the last work, for the time being, in this series
and it appears in this collection \cite{BerWey_ptrs13}.

The current paper adopts this magnetic approach to solve inverse nodal
problems on both metric and discrete graphs. For metric graphs, 'nodal
lines' are actually nodal points, the zeros of the eigenfunction,
and the nodal domains are the subgraphs bounded in between. For discrete
graphs, 'nodal lines' are edges connecting two vertices at which the
eigenvector differs in sign, and the nodal domains are the subgraphs
obtained upon removal of those edges. The main result of this paper
is a solution to an inverse nodal problem - under some genericity
assumptions, if for all $n$ the $n^{\textrm{th }}$eigenfunction
has $n-1$ zeros, then the graph cannot have cycles and must be a
tree. This result is valid for both metric and discrete graphs, although
the assumptions and proof methods of the relevant theorems are different
(see subsection \ref{sub:main_results} for exact details). We conclude
this introductory part by referring the reader to the collection of
articles, \cite{SS07}, where a broad view is given on the nodal domain
research, its history, applications and the numerous types of objects
it concerns.\\

The paper is set in the following way. In the rest of the introductory
section we familiarize the reader with all ingredients which plays
a role in the formulation of the theorems: magnetic operators on both
discrete and metric graphs, the connection between both types of graphs
and their nodal counts. The last introductory subsection \ref{sub:main_results}
presents the two main theorems of the paper. Section \ref{sec:tools_for_proofs}
introduces the three main tools needed for the proofs. The next two
sections bring the proofs of our inverse nodal theorems for both discrete
and metric graphs, each in a separate section. Section \ref{sec:discretized_vesrions}
presents a method of discretizing metric graphs which gives another
insight for the proof in section \ref{sec:inverse_theorem_metric}
and might lead to generalizations of the inverse result. Finally we
discuss the essence of the work and offer directions for further exploration.
It is important to emphasize that some of the paper's sections can
be read independently of those preceding them. The schematic diagram
in figure \ref{fig:dependency_diagram} shows the section dependencies
as blocks resting on top of other blocks on which they depend.

\begin{figure}[h]
\includegraphics[scale=0.20]{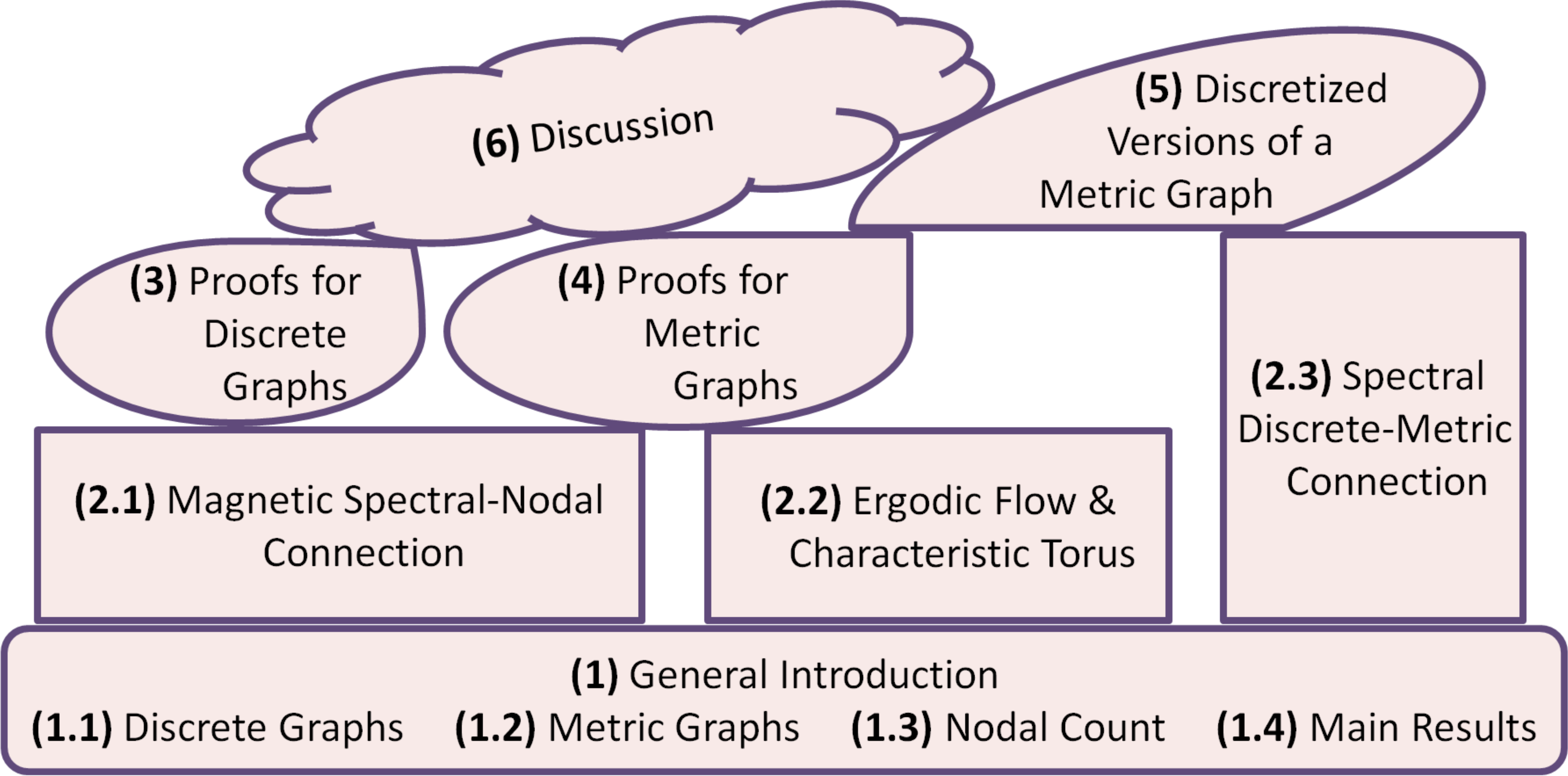}

\caption{Schematic diagram of section dependencies}
\label{fig:dependency_diagram}
\end{figure}

\subsection{Discrete graphs\label{sub:intro_discrete_graph}}

Let $\disgraph=\left(\V,\EE\right)$ be a graph with the sets of vertices
$\V$ and edges $\EE$. All graphs discussed in this paper are connected
and have a finite number of vertices and edges. Each edge $e\in\EE,\, e=\left\{ u,v\right\} $
connects a pair of vertices, $u,v\in\V$. We exclude edges which connect
a vertex to itself and do not allow two vertices to be connected by
more than one edge. For some purposes we would need to consider the
directions of the edges. We therefore also define the set $\E$ of
all graph edges, each appears twice with both its possible directions,
and hence $\left|\E\right|=2\left|\EE\right|$. We use the notation
$e=\left(u,v\right)$ to refer to a directed edge $e\in\E$ which
starts at $u$ and terminates at $v$ and denote the edge with the
reversed direction by $\hat{e}:=\left(v,u\right)$. For a vertex $v\in\V$,
its \emph{degree} equals the number of edges connected to it, i.e.,
$d_{v}:=\left|\left\{ u;\,\,\left\{ u,v\right\} \in\EE\right\} \right|$.\emph{
Functions }on the graph refer to $f:\V\rightarrow\R$ and we now present
the operators acting on such functions.

The \emph{normalized Laplacian} is
\begin{equation}
\lap_{u,v}^{(norm)}=\begin{cases}
-\nicefrac{1}{\sqrt{d_{u}d_{v}}} & \left(v,u\right)\in\E\\
0 & \left(v,u\right)\notin\E\\
1 & u=v
\end{cases}\label{eq:normalized_laplacian}
\end{equation}

The normalized Laplacian is a special case of the \emph{generalized
discrete Laplacian} also known as the \emph{discrete Schr�dinger operator},
a real symmetric matrix which obeys 
\begin{eqnarray}
\lap_{u,v}<0 &  & \textrm{if }\left\{ v,u\right\} \in\EE\nonumber \\
\lap_{u,v}=0 &  & \textrm{if }\left\{ v,u\right\} \notin\EE,\label{eq:generalized_laplacian}
\end{eqnarray}
and with no constraints on its diagonal values (which are sometimes
called in the literature on-site potentials).

These matrices have real eigenvalues which we denote by $\left\{ \lambda_{n}\right\} _{n=1}^{\left|\V\right|}$,
ordered increasingly, and the corresponding eigenvectors are denoted
by $\left\{ f_{n}\right\} _{n=1}^{\left|\V\right|}$. The spectrum
of the normalized Laplacian belongs to the interval $\left[0,2\right]$.
We refer the interested reader to many results concerning the spectra
of graph matrices and their eigenvectors which can be found in the
books \cite{Leydold_nodal,Chung_spectralgraph,CveRowSim_eigenspaces97}
and the references within.

The discrete Schr�dinger operator can be supplied with a \emph{magnetic
potential}, $\magp$, which is defined as $\magp:\E\rightarrow\R$
such that $\magp\left(e\right)=-\magp\left(\hat{e}\right)$. Assigning
such a potential to the operator amounts to the following changes

\[
\lap\left(\magp\right)\ :\ \C^{\left|\V\right|}\rightarrow\C^{\left|\V\right|}
\]

\[
\lap_{u,v}\left(\magp\right)=\lap_{u,v}\textrm{e}^{i\magp\left(u,v\right)},
\]
where $\lap_{u,v}$ are the entries of the previous (zero magnetic
potential) Laplacian. Note that this magnetic operator is still Hermitian
and hence has real eigenvalues. 

Given a cycle, $\gamma=\left(v_{1},\ldots,v_{n}\right)$, on the graph
we define the \emph{magnetic flux} through this cycle as
\[
\alpha_{\gamma}:=\left[\magp\left(v_{1},v_{2}\right)+\ldots+\magp\left(v_{n-1},v_{n}\right)+\magp\left(v_{n},v_{1}\right)\right]\,\textrm{mod}\,2\pi.
\]
The graph's cycles play an important role when introducing magnetic
potential. We denote the number of ``independent'' cycles on the
graph as 
\[
\beta:=\left|\EE\right|-\left|\V\right|+1
\]
 (assuming the graph is connected). This is also known as the first
Betti number, which is the dimension of the graph's first homology.
In the following we often fix some arbitrary $\beta$ independent
cycles on the graph, $\left\{ \gamma_{1},\ldots,\gamma_{\beta}\right\} $
and denote the magnetic fluxes through them by $\left\{ \alpha_{1},\ldots,\alpha_{\beta}\right\} $.
One can show that given two magnetic potentials, $\magp_{1}$ and
$\magp_{2}$ with the same values for the magnetic fluxes $\left\{ \alpha_{i}\right\} _{i=1}^{\beta}$,
their corresponding magnetic Laplacians, $\lap\left(\magp_{1}\right)$
and $\lap\left(\magp_{2}\right)$ are unitarily equivalent. This unitary
equivalence is also known as the gauge invariance principle and means
that the spectrum of the Laplacian is uniquely determined by the values
$\left\{ \alpha_{i}\right\} _{i=1}^{\beta}$ (but not so for its eigenfunctions).
We take advantage of this principle, notation-wise, and write the
Laplacian and its eigenvalues as functions of the magnetic parameters
(fluxes), $\lap\left(\magv\right)$ and $\lambda_{n}\left(\magv\right)$,
where $\magv:=\left(\alpha_{1},\ldots,\alpha_{\beta}\right)$. Note
that in the special case of a tree graph, the gauge invariance principle
means that all magnetic Laplacians are unitarily equivalent to the
zero magnetic Laplacian, $\lap\left(\zvec\right)$, and the eigenvalues
do not depend on the magnetic potential. More details on magnetic
operators on discrete graphs can be found in \cite{cdv_graphes98,LieLos_Duke93,Shubin_cmp94,Sunada_proceedings94}.

\subsection{Metric (quantum) graphs\label{sub:intro_metric_graph}}

A \emph{metric graph} is a discrete graph each of whose edges, $e\in\EE$,
is identified with a one dimensional interval, $[0,l_{e}]$, with
a positive finite length $l_{e}$. We use the notation $\lenv$ for
the vector whose entries are $\left\{ l_{e}\right\} _{e\in\EE}$.
We denote a metric graph by $\metgraph=\left(\V,\EE,\lenv\right)$.
We can then assign to each edge $e\in\E$ a coordinate, $x_{e}$,
which measures the distance along the edge from the starting vertex
of $e$. In particular, we have the following relation between coordinates
of reversed edges, $x_{e}+x_{\hat{e}}=l_{e}$. We denote a coordinate
by $x$, when its precise nature is unimportant. A function on the
graph is described by its restrictions to the edges, $\left\{ \left.f\right|_{e}\right\} _{e\in\E}$,
where $\left.f\right|_{e}:\left[0,l_{e}\right]\rightarrow\C$ and
we require $\left.f\right|_{e}(x_{e})=\left.f\right|_{\hat{e}}(l_{e}-x_{e})$,
for all $e\in\E$. Therefore, in effect, there is only a single function
assigned to each pair of edges, $e,\hat{e}$.

We equip the metric graphs with a self-adjoint differential operator
$\hamil$, the \emph{Hamiltonian} or \emph{metric Schr�dinger operator},
\begin{equation}
\mathcal{H}\left(\magp\right)\ :\ \left.f\right|_{e}(x_{e})\mapsto\left(\ui\frac{\ud}{\ud x_{e}}+\magp_{e}\left(x_{e}\right)\right)^{2}\left.f\right|_{e}\left(x_{e}\right)+V_{e}(x_{e})\left.f\right|_{e}\left(x_{e}\right),\label{eq:magnetic_metric_Schroedinger}
\end{equation}
where $V_{e}(x)$ is a real valued bounded and piecewise continuous
function which forms the \emph{electric potential}, and $\magp_{e}\in C^{1}\left(\left[0,l_{e}\right]\right)$
is called the \emph{magnetic potential} and it obeys $\magp_{e}\left(x_{e}\right)=-\magp_{\hat{e}}\left(x_{\hat{e}}\right)$.
It is most common to call this setting of a metric graph with a Schr�dinger
operator, a quantum graph. We will keep calling these graphs metric
graphs to distinguish them from their discrete counterpart which we
also equip here with an operator of a quantum nature.

To complete the definition of the operator we need to specify its
domain. We denote by $H^{2}(\metgraph)$ the following direct sum
of Sobolev spaces 
\begin{equation}
H^{2}(\metgraph):=\bigoplus_{e\in\E}H^{2}([0,l_{e}])\ .
\end{equation}
The graph's connectivity is expressed by matching the values of the
functions at the common vertices, thus dictating the operator's domain.
All matching conditions that lead to the operator (\ref{eq:magnetic_metric_Schroedinger})
being self-adjoint have been classified in \cite{Har_jpa00,KosSch_jpa99,Kuchment_wav04}.
It can be shown that under these conditions the spectrum of $\mathcal{H}$
is real and bounded from below \cite{Kuchment_wav04}. In addition,
since we only consider compact graphs, the spectrum is discrete and
with no accumulation points. We number the eigenvalues in the ascending
order and denote them (similarly to the discrete case) with $\left\{ \lambda_{n}\right\} _{n=1}^{\infty}$
and their corresponding eigenfunctions with $\left\{ f_{n}\right\} _{n=1}^{\infty}$.
We also use $k_{n}$, such that $\lambda_{n}=k_{n}^{2}$, and say
that $\left\{ k_{n}\right\} _{n=1}^{\infty}$ is the $k$-spectrum
of the graph.

As we wish to study the sign changes of the eigenfunctions, we would
require their continuity. The only matching conditions which assure
a self-adjoint operator and guarantee that the function is continuous
(at the vertices) are the so-called extended $\delta$-type conditions
at all the graph's vertices.

A function $f\in H^{2}(\metgraph)$ is said to satisfy the extended
$\delta$-type conditions at a vertex $v$ if
\begin{enumerate}
\item $f$ is continuous at $v\in\V$, i.e., 
\[
\forall e_{1},e_{2}\in\E_{v\,\,\,\,\,}\left.f\right|_{e_{1}}(0)=\left.f\right|_{e_{2}}(0),
\]
where $\E_{v}$ is the set of edges starting at $v$ (so that $x_{e}=0$
at $v$ if $e\in\E_{v}$). 
\item the outgoing derivatives of $f$ at $v$ satisfy 
\begin{equation}
\sum_{e\in\E_{v}}\left(\frac{\ud}{\ud x_{e}}-\ui\magp_{e}(0)\right)\left.f\right|_{e}\left(0\right)=\chi_{v}f(0),\quad\chi_{v}\in\R,\label{eq:delta_deriv_magnetic_case}
\end{equation}
 where $f(0)$ denotes the value of $f$ at the vertex (which is uniquely
defined due to the first part of the condition).
\end{enumerate}
In particular, the case $\chi_{v}=0$ is often referred to as \emph{Neumann
condition} (also called Kirchhoff or standard condition). We will
call a graph whose vertex conditions are all of Neumann type and whose
magnetic and electric potentials vanish everywhere, a \emph{Neumann
graph }and state our main results for such graphs. A Neumann graph
has $\lambda_{1}=0$ with multiplicity which equals the number of
graph's components (which is always one throughout this paper) and
their $k$-spectrum is therefore real and positive. Another useful
vertex condition is $\forall e\in\E_{v},\,\,\left.f\right|_{e}\left(0\right)=0$.
This is called a \emph{Dirichlet condition} at the vertex $v$, and
can be formally written as (\ref{eq:delta_deriv_magnetic_case}) with
$\chi_{v}=\infty$. Note that whenever a vertex exhibits a Dirichlet
condition, it effectively disconnects all edges connected to this
vertex. Similarly to the Neumann graph, a \emph{Dirichlet graph} is
obtained whenever all vertices are supplied with Dirichlet conditions.
The eigenvalues of such a graph are merely a union of spectra of its
disjoint edges (with Dirichlet conditions at their endpoints) and
are called \emph{Dirichlet eigenvalues}.

An important observation which plays a role in this paper is that
the spectrum and eigenfunctions of the graph are not affected if a
graph's edge is divided into two parts by introducing a new vertex
(of degree two) at an arbitrary point on this edge and supplying it
with Neumann conditions. We call the process (and outcome) of introduction
of any number of such vertices on the graph, graph's \emph{subdivision}
and keep in mind that the spectral properties are invariant with respect
to subdivision. Similarly to the definition of discrete graphs, we
also exclude edges which connect a vertex to itself and vertices which
are connected by more than two edges (reversed to each other). Note,
however, that this does not restrict the generality of our results
as any metric graph which does have self cycles or multiple edge can
be subdivided to eliminate these defects. 

The \emph{magnetic flux} through a cycle, $\gamma=\left(v_{1},\ldots,v_{n}\right)$,
is defined as 
\[
\alpha_{\gamma}:=\sum_{e\in\gamma}\int_{0}^{l_{e}}\magp_{e}\,\textrm{d}x_{e},
\]
where the notation $e\in\gamma$ means that either $e=\left(v_{i},v_{i+1}\right)$
for some $i$ or $e=\left(v_{n},v_{1}\right)$. The gauge invariance
principle introduced previously applies to the metric Schr�dinger
operator as well and allows us to write $\mathcal{H}\left(\magv\right)$
and $\lambda_{n}\left(\magv\right)$, when referring to the operator
and its eigenvalues. This obviously holds once some fixed choice $\beta$
cycles is made, and the notation $\left\{ \alpha_{i}\right\} _{i=1}^{\beta}$
is adapted to fluxes through these cycles, where $\magv:=\left(\alpha_{1},\ldots,\alpha_{\beta}\right)$.
Two good references for further reading on the general theory of metric
(quantum) graphs are \cite{BerKuc_quantum_graphs13,GS06}.

\subsection{The nodal count}

The main focus of this paper is the number of sign changes of eigenfunctions
on discrete and metric graphs. When counting sign changes we always
consider eigenfunctions of the zero magnetic potential operators,
as otherwise we are not guaranteed to have real valued eigenfunctions.
In addition, in order for the number of sign changes to be well defined
we assume the following.

\begin{assumption}\label{ass:generality_assumption}

The eigenvalue $\lambda_{n}$ is simple and the corresponding eigenfunction,
$f_{n}$ is different than zero on every vertex.

\end{assumption}We call $\lambda_{n}$ a \emph{generic eigenvalue}
if it satisfies assumption \ref{ass:generality_assumption} (correspondingly,
$f_{n}$ is called a \emph{generic }eigenfunction). This assumption
is generic with respect to various perturbations to the operator.
In the discrete case such perturbations include changing the non-zero
entries of the matrix $\lap$. For a metric graph, one may perturb
either the vertex conditions, the edge lengths or the electric potential.
Further discussions can be found in \cite{BanBerRazSmi_cmp12,Ber_arx11,Fri_ijm05},
where this assumption was used.

We now define sign changes (also known as nodal points) and nodal
domains and use the same notations for both metric and discrete graphs.
\begin{defn}
~
\begin{enumerate}
\item Let $f_{n}\in H^{2}(\metgraph)$ be generic and the $n^{\textrm{th}}$
eigenfunction of the Schr�dinger operator on a metric graph. The zeros
of $f_{n}$ form isolated points on the graph and they correspond
to the function's sign changes. Their set is denoted by $\Phi_{n}$.
\item Let $f_{n}\in\R^{\V}$ be generic and the $n^{\textrm{th}}$ eigenfunction
of the discrete Schr�dinger operator on a graph. We say that an edge
$e=\left\{ u,v\right\} $ forms a \emph{sign change (also nodal point)}
of the eigenfunction if $f_{n}\left(u\right)f_{n}\left(v\right)<0$.
The set of these edges is denoted by 
\[
\Phi_{n}:=\left\{ \left\{ u,v\right\} \in\EE\,|\,\, f_{n}\left(u\right)f_{n}\left(v\right)<0\right\} .
\]

\item The number of nodal points of $f_{n}$ (for both metric and discrete
graphs) is called the \emph{sign change count} or \emph{nodal point
count} of $f_{n}$ and is denoted by $\phi_{n}:=\left|\Phi_{n}\right|$.
\item $\metgraph\backslash\Phi_{n}$ ($\disgraph\backslash\Phi_{n}$) consists
of a few subgraphs disconnected one from the other. These subgraphs
are called the \emph{nodal domains} of $f_{n}$ and their number,
the \emph{nodal domain count}, is denoted by $\nc$. 
\end{enumerate}
\end{defn}
We also use the general term \emph{nodal count}, when it is either
clear or unimportant whether we count nodal points or nodal domains.
We remark that if assumption \ref{ass:generality_assumption} is not
satisfied due to zeros on vertices, the definitions above should be
modified. There are indeed alternative definitions of the nodal count
in such scenarios (see \cite{BOS08}), but their treatment is out
of the scope of this paper.

The known bounds on the nodal count are

\begin{eqnarray}
n-1 & \leq\fc\leq & n-1+\beta\label{eq:bounds_on_zero_count}\\
n-\beta & \leq\nc\leq & n.\label{eq:bounds_on_nodal_count}
\end{eqnarray}

In particular, for a tree graph where $\beta=0$, one obtains $\forall n\,\,\fc=n-1$.
This result for the interval is the famous Sturm's oscillation theorem
\cite{Stu_jmpa36,Stu_jmpa36a} and its generalization for trees was
done in \cite{PokPryObe_mz96,Schapotschnikow06}. The fact that \emph{discrete}
tree graphs also have this nodal count is proven in \cite{Berkolaiko07,DhaRam_prl85,Fiedler_cmj75}.
The upper bound of (\ref{eq:bounds_on_nodal_count}) (Courant bound)
is proven in \cite{DGLS01} for discrete graphs and in \cite{GnuSmiWeb_wrm03}
for metric graphs. For both metric and discrete graphs, the upper
bound of (\ref{eq:bounds_on_nodal_count}) also proves the upper bound
of (\ref{eq:bounds_on_zero_count}) since $\fc\leq\nc-1+\beta$. The
lower bound of (\ref{eq:bounds_on_nodal_count}) for both metric and
discrete graphs is proven in \cite{Berkolaiko07} and the same proof
works to prove the lower bound in (\ref{eq:bounds_on_zero_count})
(again for both kinds of graphs).

Two recent works which go further than the above bounds (and from
which the bounds (\ref{eq:bounds_on_zero_count}), (\ref{eq:bounds_on_nodal_count})
can be deduced) characterize the nodal count of an eigenfunction in
terms of a Morse index of a predefined energy function, \cite{BanBerRazSmi_cmp12,BerRazSmi_jpa12}.
These works led to a characterization of the nodal count in terms
of Morse indices of magnetic perturbations. This last result is an
important tool used in the proofs of the current paper and is described
in section \ref{sub:magnetic_spectral_nodal_connection}.

\subsection{\label{sub:main_results}The main results of the current paper}

We start by introducing the notation

\[
\nums:=\left\{ n;\,\,\lambda_{n}\,\textrm{satisfies assumption }\ref{ass:generality_assumption}\right\} .
\]
Namely, $\lambda_{n}$ is a \emph{generic eigenvalue} if $n\in\nums$.

The two main theorems of this paper are
\begin{thm}
\label{thm:main_theorem_discrete}Let $\disgraph$ be a graph supplied
with discrete Schr�dinger operator, all of whose eigenvalues are generic,
i.e., $\nums=\left\{ 1,\ldots,\left|\V\right|\right\} $. If its nodal
point count is such that for all $n\in\nums$,~$\phi_{n}=n-1$ then
$\disgraph$ is a tree graph.
\begin{thm}
\label{thm:main_theorem_metric}Let $\metgraph$ be a Neumann metric
graph with at least one generic eigenvalue greater than zero. Then
$\metgraph$ has infinitely many generic eigenvalues and exactly one
of the following holds\end{thm}
\begin{enumerate}
\item $\metgraph$ is a tree and for all $n\in\nums,\,\phi_{n}=n-1$ and
$\nc=n$.
\item $\metgraph$ is not a tree and each of the sets $\left\{ n\in\nums\,;\,\fc>n-1\right\} $,
$\left\{ n\in\nums\,;\,\nc<n\right\} $ is infinite.
\end{enumerate}
\end{thm}
These theorems solve the nodal inverse problem for a tree graph. It
is interesting to compare the assumptions and the conclusions of the
discrete case with the metric one. The metric theorem shows that a
generic nodal count cannot be almost like this of a tree - either
it equals a tree's nodal count or it differs from it for an infinite
subsequence. In addition, it allows to conclude that the graph is
a tree even if non-generic eigenvalues exist (by simply ignoring them)
as long as there is at least one positive generic eigenvalue. The
discrete theorem seems to assume more than the metric one, as it requires
that all the eigenvalues are generic and all of them have the nodal
count of a tree ($\phi_{n}=n-1$). Indeed, in the discrete case, ignoring
non-generic eigenvalues is not possible - the normalized Laplacian
on the triangle graph, for example, has a single generic eigenfunction
with the tree nodal count and two others which are non-generic. The
metric theorem also allows to conclude that the graph is a tree based
on its \emph{nodal domain} count. There is, however, no analogous
result for discrete graphs. As the final comparison point we mention
that the metric theorem solves the inverse problem only for the Laplacian
(as it is a Neumann graph and there is no potential) whereas the discrete
one applies to the more general Schr�dinger operator.

\section{Introducing Tools Needed for the Proofs\label{sec:tools_for_proofs}}

\subsection{The magnetic spectral-nodal connection\label{sub:magnetic_spectral_nodal_connection}}

This subsection is devoted to an important connection between the
nodal count and the stability of eigenvalues under magnetic perturbation.
Such a connection first appeared in \cite{Ber_arx11}, where Berkolaiko
proved it for discrete graphs. Shortly afterward, the same theorem
was reproved by Colin de Verdi\`{e}re who had also shown that the
theorem holds for the Hill operator - the metric Schr�dinger operator
on a single loop graph, \cite{cdv_arx12}. The proof of the theorem
for a general metric graph, due to Berkolaiko and Weyand, appears
in another manuscript of this same issue, \cite{BerWey_ptrs13}.

Before stating the relevant theorem, we need to introduce the following
notations.
\begin{enumerate}
\item The difference of the nodal count from its lower bound (also its value
in the tree case) is denoted
\[
\sur:=\fc-\left(n-1\right).
\]
Hence, we call it the \emph{nodal surplus} following \cite{Ber_arx11}
(it was also called \emph{nodal defect} in \cite{cdv_arx12}).
\item The Hessian of the eigenvalue $\lambda$ with respect to magnetic
parameters at the point $\magv=\zvec$ is denoted by 
\[
\left[\hes_{\lambda}\left(\zvec\right)\right]_{i,j}:=\left.\frac{\partial^{2}\lambda}{\partial\alpha_{i}\partial\alpha_{j}}\right|_{\magv=\zvec}.
\]
When $\magv=\zvec$ is a critical point of $\lambda\left(\magv\right)$,
the Morse index is defined as the number of negative eigenvalues of
$\hes_{\lambda}\left(\zvec\right)$ and it is denoted by $\morse_{\lambda}\left(\zvec\right)$.
\end{enumerate}
We collect below the main results of the papers \cite{Ber_arx11,BerWey_ptrs13,cdv_arx12}
into a single theorem which is applicable for both metric and discrete
graphs.
\begin{thm}
\label{thm:magnetic_morse_index}\cite{Ber_arx11,BerWey_ptrs13,cdv_arx12}
Let $\disgraph$ $\left(\metgraph\right)$ be a discrete (metric)
graph supplied with a\emph{ }discrete (metric) Schr�dinger operator.
Let $\lambda_{n}\left(\magv\right)$ and $f_{n}\left(\magv\right)$
be an eigenvalue and a corresponding eigenfunction of the magnetic
operator such that $\lambda_{n}\left(\zvec\right)$ and $f_{n}\left(\zvec\right)$
satisfy assumption \ref{ass:generality_assumption}. Then the following
holds.
\begin{enumerate}
\item \label{enu:magnetic_morse_index_part_1}The point $\magv=\zvec$ is
a critical point of the function $\lambda_{n}\left(\magv\right)$.
\item \label{enu:magnetic_morse_index_part_2}The critical point $\magv=\zvec$
is non-degenerate.
\item \label{enu:magnetic_morse_index_part_3}The nodal surplus, $\sigma_{n}$,
of the eigenfunction $f_{n}\left(\zvec\right)$ is equal to the Morse
index of this critical point, 
\[
\sur=\morse_{\lambda_{n}}\left(\zvec\right).
\]

\end{enumerate}
\end{thm}

\subsection{Ergodic flow on the characteristic torus\label{sub:ergodic_motion_on_the_graph_torus}}

It is well known that eigenvalues of a metric graph are given, with
their multiplicities, as the zeros of a secular function \cite{KS97,KotSmi_ap99},
\[
\left\{ k^{2};\,\,\tilde{F}\left(k;\,\lenv;\,\magv\right)=0\right\} ,
\]
where

\begin{equation}
\tilde{F}\left(k;\,\lenv;\,\magv\right):=\det\left(\ue^{-\nicefrac{\ui}{2}\left(\mathbf{\magmat}\left(\magv\right)+k\mathbf{\lenmat}\left(\lenv\right)\right)}\right)\det\left(\scatmat\left(k\right)\right)^{-\nicefrac{1}{2}}\det\left(\id-\ue^{\ui\left(\mathbf{\magmat}\left(\magv\right)+k\mathbf{\lenmat}\left(\lenv\right)\right)}\scatmat\left(k\right)\right),\label{eq:secular_function}
\end{equation}
and $\magmat,\,\lenmat\,\textrm{\,\ and }\scatmat$ are square matrices
of dimension $\left|\E\right|$ and contain the information about
the magnetic fluxes, edge lengths and edge connectivity, respectively.
Exact details on the structure of those matrices appear in \cite{GS06,KotSmi_ap99},
and we just state here the necessary facts which we use later on:
\begin{enumerate}
\item $\lenmat$ is a diagonal matrix of the form 
\[
\lenmat=\textrm{diag}\left\{ l_{e}\right\} _{e\in\E}.
\]

\item For a Neumann graph $\scatmat$ is a constant unitary matrix which
does not depend on $k$.
\item $\magmat$ is a diagonal matrix which is linear in $\magv$.
\item $\tilde{F}$ is a real valued function.
\end{enumerate}
Let $\metgraph=\left(\V,\EE,\lenv\right)$ be a metric graph. We write
its edge lengths as the following linear combinations.
\begin{equation}
\forall e\in\EE\,\,\, l_{e}=\sum_{i\in I}r_{i}^{\left(e\right)}\xi_{i},\label{eq:edge_lengths_as_linear_combinations}
\end{equation}
where all $r_{i}^{\left(e\right)}$ are rational numbers and $\left\{ \xi_{i}\right\} _{i\in I}$
is a set of \emph{incommensurate} real numbers, i.e., they are linearly
independent over the rationals. For example, if all edge lengths are
incommensurate then $\left|I\right|=\left|\EE\right|$ and the set
$\left\{ \xi_{i}\right\} _{i\in I}$ can be chosen to consist of the
edge lengths. On the other extreme, if all ratios of edge lengths
are rational, then $\left|I\right|=1$ and one can choose $\xi_{1}$
to equal any of the edge lengths (or any rational multiple of it).
The relations between the edge lengths of the graph, $\left\{ l_{e}\right\} _{e\in\EE}$,
and the parameters $\left\{ \xi_{i}\right\} _{i\in I}$ will play
an important role later on and we thus define the \emph{length map
}as 
\[
\lenmap:\R^{I}\rightarrow\R^{\EE}
\]

\begin{equation}
\left[\lenmap\left(\tvec\right)\right]_{e}=\sum_{i\in I}r_{i}^{\left(e\right)}x_{i}.\label{eq:length_map}
\end{equation}
This map depends on the specific edge lengths of $\metgraph$ (even
the dimension of its domain depends on that) and on the specific choice
of values for $\left\{ \xi_{i}\right\} _{i\in I}$. 

We now describe a method introduced by Barra and Gaspard \cite{BarGas_jsp00}
who related the graph eigenvalues to the Poincar\'{e} return times
of a flow to a surface defined by the zero level set of the secular
function (\ref{eq:secular_function}). We present their method using
our length map and start by redefining the secular function

\begin{eqnarray}
F\,:\,\R^{I}\times\R^{\beta} & \rightarrow & \R\\
F\left(\tvec;\,\magv\right) & := & \tilde{F}\left(1;\,\lenmap\left(\tvec\right);\,\magv\right).\label{eq:secular_equation_defined_on_torus}
\end{eqnarray}

We point out some properties of $F$, which can be deduced from (\ref{eq:secular_function}).
\begin{enumerate}
\item $F$ is differentiable.
\item \label{enu:secular_function_second_property}$F\left(k\vec{\xi};\,\magv\right)=\tilde{F}\left(k;\,\lenmap\left(\vec{\xi}\right);\,\magv\right)$
since $\lenmap$ is homogeneous and $\tilde{F}$ contains the parameter
$k$ only in the product $k\lenv=k\lenmap\left(\vec{\xi}\right)$.
\item $F\left(\tvec;\,\magv\right)$ is periodic in each of the entries
of $\tvec$. The period depends on the specific entry, $x_{i}$, and
is some rational multiple of $2\pi$.
\end{enumerate}
The last property allows us to (re)define the function $F\left(\cdot;\,\magv\right)$
on an $\left|I\right|$-dimensional torus, $\torus$, with sides depending
on the periodicity of $F$ with respect to its parameters, $\left\{ x_{i}\right\} _{i\in I}$.
Namely, 
\[
F\,:\,\torus\times\R^{\beta}\rightarrow\R,
\]
and from now on whenever $F$ is mentioned, its $\tvec$ variable
is taken modulus this torus periodicity, even if this is not explicitly
written. Property (\ref{enu:secular_function_second_property}) allows
us to characterize the graph $k$-eigenvalues as 
\[
\left\{ k\left(\magv\right);\,\, F\left(k\vec{\xi};\,\magv\right)=0\right\} ,
\]
where $\vec{\xi}=\lenmap^{-1}\left(\lenv\right)$ is the vector with
incommensurate entries chosen above. We therefore may define the following
flow on the torus
\begin{equation}
\tvec\left(k\right):=k\vec{\xi},\label{eq:flow_definition}
\end{equation}
and the surface 
\begin{equation}
\Sigma_{\magv}:=\left\{ F\left(k\vec{\xi};\,\magv\right)=0\right\} ,\label{eq:level_set_of_secular_function}
\end{equation}
so that the $k$-spectrum equals the times (i.e., the $k$ values)
for which the flow $\tvec\left(k\right)$ pierces $\Sigma_{\magv}$.
\begin{rem}
\label{rem:eigenfunctions_on_the_torus}The eigenfunctions which correspond
to the graph eigenvalues depend both on the point on the torus, $k\vec{\xi}\in\torus$,
and on the specific values of $k$ and $\lenv$. The restriction of
the eigenfunction to the graph vertices, however, is a continuous
function solely of $k\vec{\xi}\in\torus$ (see e.g., \cite{KS97,KotSmi_ap99}).
This last observation is exploited in sections \ref{sec:inverse_theorem_metric}
and \ref{sec:discretized_vesrions}.
\end{rem}
We end by noting that as the entries of $\vec{\xi}$ are linearly
independent over $\Q$, the flow (\ref{eq:flow_definition}) is ergodic
on $\torus$. This ergodicity is the reason for making the reduction
from the set of edge lengths, $\left\{ l_{e}\right\} _{e\in\EE}$,
to $\left\{ \xi_{i}\right\} _{i\in I}$. We could have defined the
flow on an $\left|\EE\right|$-dimensional torus, but then the flow
would not necessarily fill the whole torus - it would actually fill
a linear subspace given by $\lenmap\left(\torus\right)$ (modulo the
$\left|\EE\right|$-dimensional torus).

\subsection{The spectral connection between metric and discrete graphs\label{sub:discrete_metric_connection}}

An interesting connection exists between the spectrum of an \emph{equilateral}
metric graph,\emph{ }a graph whose all edge lengths are equal, and
the spectrum of the discrete graph which shares the same connectivity.
This connection is usually stated in terms of the spectrum of the
\emph{transition matrix}, 
\begin{equation}
\trans_{u,v}=\begin{cases}
\nicefrac{1}{d_{u}} & \left(v,u\right)\in\E\\
0 & \left(v,u\right)\notin\E
\end{cases}.\label{eq:transition_matrix}
\end{equation}
The transition matrix is sometimes referred to as the \emph{difference
operator} or even the discrete Laplace operator, but we will not use
it here to avoid confusion. One should note that the transition matrix
is not symmetric (and thus does not fall under the definition of the
generalized Laplacian), but it is closely related to the normalized
Laplacian since
\begin{equation}
\trans=\id-\degmat^{\nicefrac{1}{2}}\lap^{\left(norm\right)}\degmat^{-\nicefrac{1}{2}},\label{eq:transition_and_normalized_laplacian}
\end{equation}
where $\degmat=\textrm{diag}\left\{ d_{v}\right\} _{v\in\V}$ is a
diagonal matrix which contains all vertex degrees. If $\lambda,f$
are an eigenvalue and an eigenvector of $\lap^{\left(norm\right)}$,
then $\trans$ has $1-\lambda$ and $\degmat^{\nicefrac{1}{2}}f$
as its corresponding eigenpair. We use this connection between both
operators to slightly rephrase a known result which usually refers
to the spectrum of the transition matrix.
\begin{thm}
\label{thm:discrete_metric_spectral_connection} \cite{Cat_mm97,Kuchment_wav04,Pan_lmp06,Bel_laa85,BelMug_arx12}
Let $\disgraph$ be a discrete graph and $\metgraph$ be a metric
graph with the same connectivity and such that $\forall e\,\, l_{e}=1$.
Consider the normalized Laplacian, $\lap^{\left(norm\right)}$, on
$\disgraph$ and the metric Laplacian with Neumann conditions on $\metgraph$.
Equip both operators with magnetic fluxes for some choice of $\beta$
cycles on the graph (the same choice for both $\disgraph$ and $\metgraph$).
Let $\mu\left(\magv\right)\notin\left\{ 0,2\right\} $ be an eigenvalue
of $\disgraph$ and $f\left(\magv\right)$ the corresponding eigenvector.
Then
\begin{enumerate}
\item All values in the infinite set $\left\{ \left(\arccos\left[1-\mu\left(\magv\right)\right]\right)^{2}\right\} $
are eigenvalues of the magnetic metric Schr�dinger operator on $\metgraph$.
We consider here $\arccos$ as a multivalued function, $\arccos:\left[-1,1\right]\rightarrow\left[0,\infty\right)$.
\item When $\magv=\zvec$, the other eigenvalues of the metric Schr�dinger
operator are the Dirichlet eigenvalues, all of which belong to the
set $\left\{ \left(\pi n\right)^{2}\right\} _{n\in\Z}$. 
\item An eigenfunction corresponding to any of the eigenvalues $\left(\arccos\left[1-\mu\left(\magv\right)\right]\right)^{2}$
equals to $\degmat^{\nicefrac{1}{2}}f\left(\magv\right)$ when restricted
to $\metgraph$'s vertices.\label{enu:discrete_metric_thm_eigenfunction_on_vertices} 
\end{enumerate}
\end{thm}
The content of the theorem for the zero magnetic potential case appears
in \cite{Cat_mm97,Kuchment_wav04,Bel_laa85,BelMug_arx12}, where they
mostly treat Neumann vertex conditions (except in \cite{BelMug_arx12}
where the so called anti-Kirchhoff conditions are treated as well).
A more general derivation which includes electric and magnetic potentials
as well as $\delta$-type conditions appears in \cite{Pan_lmp06}.
In most works above the results are stated in terms of the transition
matrix, (\ref{eq:transition_matrix}), but we prefer to relate to
the normalized Laplacian as theorem \ref{thm:magnetic_morse_index}
applies to it.

\section{Proofs for Discrete Graphs\label{sec:inverse_theorem_discrete}}

The next theorem brings one of the two results which inspire this
paper's title. It is stated and proved here for discrete graphs, whereas
the metric case appears in the next section.
\begin{proof}
[Proof of theorem  \ref{thm:main_theorem_discrete}] Assume by contradiction
that $\disgraph$ is not a tree. Namely $\beta>0$, and we may therefore
supply the graph with a magnetic potential and apply theorem \ref{thm:magnetic_morse_index}.
From the assumption in our theorem we get for the nodal surplus 
\[
\forall n\,\,\sur=\phi_{n}-\left(n-1\right)=0.
\]
We apply theorem \ref{thm:magnetic_morse_index} for $\lambda_{n}\left(\magv\right)$
- the third part of the theorem allows to conclude that $\lambda_{n}\left(\magv\right)$
has a minimum at $\magv=\zvec$ and the second part shows that this
minimum is strict. Therefore, the eigenvalue sum, $\sum_{n}\lambda_{n}\left(\magv\right)$,
also has a strict minimum at $\magv=\zvec$ . However, since the diagonal
entries of the Laplacian do not depend on the magnetic parameters,
$\magv$, we get that $\textrm{trace}\left\{ \lap\left(\magv\right)\right\} $
is a constant function of $\magv$. Hence we arrive at a contradiction,
due to $\sum_{n}\lambda_{n}\left(\magv\right)=\textrm{trace}\left\{ \lap\left(\magv\right)\right\} $.
\end{proof}
The essence of the proof above is to show that the zero sequence is
not a valid candidate as a nodal surplus sequence. This is done by
identifying $\textrm{trace}\left\{ \lap\left(\magv\right)\right\} $
as a spectral invariant independent of the magnetic potential. We
wish to point out similar results which are obtained from such a method.
An immediate next step would be to observe that $\textrm{trace}\left\{ \lap^{2}\left(\magv\right)\right\} $
is a constant function of $\magv$ as well. Computing its Hessian
and expressing it in terms of the eigenvalues and their Hessians,
$\hes_{\lambda_{n}}$, gives
\begin{equation}
\sum_{n=1}^{\left|\V\right|}\lambda_{n}\left(\zvec\right)\hes_{\lambda_{n}}\left(\zvec\right)=\zmat,\label{eq:combination_of_hessian_equals_zero_2}
\end{equation}
where we have used $\forall n\,\forall i\,\,\frac{\partial}{\partial\alpha_{i}}\lambda_{n}\left(\zvec\right)=0$,
which we get from theorem \ref{thm:magnetic_morse_index} (part (\ref{enu:magnetic_morse_index_part_1})).
Similarly, the magnetic invariance of $\textrm{trace}\left\{ \lap\left(\magv\right)\right\} $
gives 
\begin{equation}
\sum_{n=1}^{\left|\V\right|}\hes_{\lambda_{n}}\left(\zvec\right)=\zmat.\label{eq:combination_of_hessian_equals_zero_1}
\end{equation}
Combining (\ref{eq:combination_of_hessian_equals_zero_1}) and (\ref{eq:combination_of_hessian_equals_zero_2})
allows to prove the following.
\begin{prop}
\label{thm:impossible_surplus_seq_1-1}Let $\disgraph$ be a graph
with $\beta$ cycles supplied with discrete Schr�dinger operator such
that all of its eigenvalues are generic. Its nodal count cannot be
of the form 
\[
\phi_{n}=\begin{cases}
n-1 & \,\,\, n\leq m\\
n-1+\beta & \,\,\, n>m
\end{cases}
\]
for any $m$.\end{prop}
\begin{proof}
Assume by contradiction that the graph has the above nodal count.
Namely, the graph's surplus sequence is of the form $\left\{ 0,\ldots,0,\beta,\ldots,\beta\right\} $.
From theorem \ref{thm:magnetic_morse_index} we get that 
\begin{eqnarray}
\forall\,\,1\leq n\leq m &  & \hes_{\lambda_{n}}\left(\zvec\right)>\zmat\label{eq:positivity_of_hessians}\\
\forall\,\, m+1\leq n\leq\left|\V\right| &  & \hes_{\lambda_{n}}\left(\zvec\right)<\zmat\label{eq:negativity_of_hessians}
\end{eqnarray}
where the first (second) inequality above means that the Hessians
are positive (negative) definite quadratic forms. Note that $m<\left|\V\right|$
due to theorem \ref{thm:main_theorem_discrete}. We rewrite (\ref{eq:combination_of_hessian_equals_zero_2})~as
\begin{eqnarray*}
\zmat & = & \sum_{n=1}^{m}\lambda_{n}\left(\zvec\right)\hes_{\lambda_{n}}\left(\zvec\right)+\sum_{n=m+1}^{\left|\V\right|}\lambda_{n}\left(\zvec\right)\hes_{\lambda_{n}}\left(\zvec\right)\\
 & < & \lambda_{m+1}\left(\zvec\right)\sum_{n=1}^{m}\hes_{\lambda_{n}}\left(\zvec\right)+\lambda_{m+1}\left(\zvec\right)\sum_{n=m+1}^{\left|\V\right|}\hes_{\lambda_{n}}\left(\zvec\right)\\
 & = & \lambda_{m+1}\left(\zvec\right)\sum_{n=1}^{\left|\V\right|}\hes_{\lambda_{n}}\left(\zvec\right),
\end{eqnarray*}
where we used (\ref{eq:positivity_of_hessians}), (\ref{eq:negativity_of_hessians})
and the fact that eigenvalues are ordered increasingly to get the
second line. The last line is proportional to (\ref{eq:combination_of_hessian_equals_zero_1}),
which gives the contradiction.
\end{proof}
Carrying on with this route and examining traces of higher powers
of the operator, shows that in general $\textrm{trace}\left\{ \lap^{k}\left(\magv\right)\right\} $
might depend on the magnetic parameters. More specifically, a direct
calculation of $\textrm{trace}\left\{ \lap^{k}\left(\magv\right)\right\} $
shows that it can be expressed as an expansion over closed walks of
size $k$ on the graph. Examining the dependence of such walks on
magnetic parameters brings about the following.
\begin{thm}
\label{thm:shortest_cycle}Let $\disgraph$ be a non-tree graph ($\beta>0$)
with discrete Schr�dinger operator such that all of its eigenvalues
are generic. Then the size of the graph's shortest cycle (its girth)
is 
\[
\min_{k}\left\{ \sum_{n}\lambda_{n}^{k-1}\hes_{\lambda_{n}}\neq\zmat\right\} .
\]
Alternatively, it also equals \textup{$\min_{k}\left\{ \sum_{n}\lambda_{n}^{k-1}\textrm{trace}\hes_{\lambda_{n}}\neq0\right\} $.}\end{thm}
\begin{proof}
The Hessian of $\textrm{trace}\left\{ \lap^{k}\right\} $ is obtained
in terms of Hessians of the eigenvalues as 
\begin{equation}
\hes{}_{\textrm{trace}\left\{ \lap^{k}\right\} }\left(\zvec\right)=\sum_{n=1}^{\left|\V\right|}\lambda_{n}^{k-1}\hes_{\lambda_{n}}\left(\zvec\right),\label{eq:hessian_of_trace_power}
\end{equation}
where we used theorem \ref{thm:magnetic_morse_index} to conclude
that $\lambda_{n}$ has a critical point at $\magv=\zvec$, and therefore
no first derivatives of $\lambda_{n}$ appear above. The diagonal
entries of $\lap^{k}$ can be expressed using closed walks on the
graph. We define a \emph{closed walk} by $\gamma=\left(v_{0},v_{1},\ldots,v_{k-1}\right)$,
where for all $0\leq i\leq k-1$ either $v_{i}=v_{i+1}$ or $\left(v_{i},v_{i+1}\right)\in\E$
(we denote $v_{k}:=v_{0}$ and keep in mind that $\gamma$ is defined
up to cyclic permutations). This is similar to the usual notion of
closed walks, but we allow the walk to stop at any vertex for a (discrete)
while before continuing to the next one. The set of all closed walks
of size $k$, passing through vertex $v$, is denoted $\mathcal{C}_{v}^{\left(k\right)}$
and we attribute to each walk a weight obtained as a product of all
corresponding Laplacian entries, 
\[
L_{\gamma}:=\lap_{v,v_{k-1}}\ldots\lap_{v_{2},v_{1}}\lap_{v_{1},v},
\]
where the notation $v:=v_{0}$ is implied. The expression above depends
on the magnetic parameters, which we omit for the sake of brevity.
The diagonal entries of $\lap^{k}$ can be now written as 
\[
\left[\lap^{k}\right]_{v,v}=\sum_{\gamma\in\mathcal{C}_{v}^{\left(k\right)}}L_{\gamma}.
\]
Denote by $\kappa$ the graph's girth. For $k<\kappa$, the closed
walks $\gamma\in\mathcal{C}_{v}^{\left(k\right)}$ do not circulate
any cycle and therefore their weights, $L_{\gamma}$, are independent
of magnetic parameters (as can also be seen by a direct calculation
of $L_{\gamma}$). We thus get from (\ref{eq:hessian_of_trace_power})
that $\min_{k}\left\{ \sum_{i}\lambda_{i}^{k-1}\hes_{\lambda_{i}}\neq\zmat\right\} \geq\kappa$
and it is left to show that this is actually an equality. For $k=\kappa$,
let $\gamma=\left(v,v_{1},\ldots,v_{k-1}\right)$ be a closed walk
on the graph which contains one of the graph's cycles. The walk $\gamma$
must have all of its vertices different from each other, as $k$ is
the length of the shortest cycle on the graph. The contribution to
$\left[\lap^{k}\right]_{v,v}$ comes only from $\gamma$ and other
walks which circulate one of the graph cycles. We may couple all such
walks to 
\[
\gamma=\left(v,v_{1},\ldots,v_{k-1}\right)\,\textrm{and \,\ }\hat{\gamma}=\left(v,v_{k-1},\ldots,v_{1}\right),
\]
and get that their contributions are complex conjugates of each other,
$L_{\gamma}=\overline{L_{\hat{\gamma}}}$. From $\lap\left(\magv\right)$
being Hermitian and $\lap\left(\zvec\right)$ having negative off-diagonal
entries, we get 
\[
L_{\gamma}+L_{\hat{\gamma}}=\left(-1\right)^{k}\left|L_{\gamma}\right|2\cos\left(\vec{n}\cdot\magv\right),
\]
where $\vec{n}\in\Z^{\beta}$. As $\gamma$ circulates one of the
graph cycles, $\vec{n}\neq\zvec$ and therefore
\[
\exists i\,\,\textrm{s.t. \,}\left.\frac{\partial^{2}}{\partial\alpha_{i}^{2}}\left(L_{\gamma}+L_{\hat{\gamma}}\right)\right|_{\magv=\zvec}\neq0.
\]
In particular, all such second derivatives which do not vanish have
a definite sign, which equals to $\left(-1\right)^{k+1}$. Therefore,
summing over all such couples, $\gamma,\hat{\gamma}$ gives 
\[
\textrm{sign}\left\{ \left.\frac{\partial^{2}}{\partial\alpha_{i}^{2}}\left(\left[\lap^{k}\right]_{v,v}\right)\right|_{\magv=\zvec}\right\} =\left(-1\right)^{k+1},
\]
and therefore also 
\[
\textrm{sign}\left\{ \left.\frac{\partial^{2}}{\partial\alpha_{i}^{2}}\left(\textrm{trace}\left\{ \lap^{k}\right\} \right)\right|_{\magv=\zvec}\right\} =\left(-1\right)^{k+1}.
\]
This shows that for $k=\kappa$, the trace of the Hessian in (\ref{eq:hessian_of_trace_power})
is different than zero and completes the proof.\end{proof}
\begin{rem}
Note that the proof would work similarly if the second derivatives
are calculated with respect to the magnetic potential on the single
edges (rather then the flux over a cycle). In such a case, the theorem
can be extended to include tree graphs as well. In addition, these
derivatives are more accessible for computation, especially we are
only given the Laplacian and the graph's connectivity is unknown,
which is relevant as we are dealing with inverse problems.
\end{rem}
Theorem \ref{thm:shortest_cycle} goes a step forward from theorem
\ref{thm:main_theorem_discrete} as it allows to obtain some information
on the graph's cycles. Yet, the information used in this inverse result
is purely spectral - eigenvalues and their perturbations with respect
to magnetic potentials. It would be interesting to see if one may
obtain results of similar character from the nodal count of the graph.
A possible direction might be to examine the magnetic derivatives
of traces of powers of the Laplacian (as in the proof of theorem \ref{thm:shortest_cycle})
or the magnetic derivatives of the coefficients of the characteristic
polynomial. Theorem \ref{thm:magnetic_morse_index} would probably
be a main tool in such an exploration. Indeed, in the course of the
proof of theorem \ref{thm:shortest_cycle} we did not exploit the
full strength of theorem \ref{thm:magnetic_morse_index} and used
it only to claim that all first magnetic derivatives vanish. Applying
the spectral-nodal connection which theorem \ref{thm:magnetic_morse_index}
offers to gather information on the graph cycles, might be an important
step in developing a trace formula for the nodal count. A further
discussion on this direction is found in section \ref{sec:discussion}.

\section{Proofs for Metric Graphs\label{sec:inverse_theorem_metric}}

The main result of this paper for metric graphs (theorem \ref{thm:main_theorem_metric})
follows from the next lemma and theorem.
\begin{lem}
\label{lem:infinitely_many_generic_eigenvalues}Let $\metgraph$ be
a Neumann metric graph with at least one generic eigenvalue greater
than zero. Then $\metgraph$ has infinitely many generic eigenvalues.\end{lem}
\begin{thm}
\label{thm:infinitely_many_with_same_or_symmetric_surplus}Let $\metgraph$
be a Neumann metric graph with $\beta>0$ cycles and at least one
generic eigenvalue greater than zero. Denote by $\sigma$ the nodal
surplus of such an eigenvalue. Then
\begin{enumerate}
\item there are infinitely many generic eigenvalues whose nodal surplus
equals $\sigma$.
\item there are infinitely many generic eigenvalues whose nodal surplus
equals $\beta-\sigma$.
\end{enumerate}
\end{thm}
We first apply lemma \ref{lem:infinitely_many_generic_eigenvalues}
and theorem \ref{thm:infinitely_many_with_same_or_symmetric_surplus}
to prove theorem \ref{thm:main_theorem_metric}, and then proceed
to prove them as well.
\begin{proof}
[Proof of theorem \ref{thm:main_theorem_metric}]

Observe that $\metgraph$ has infinitely many generic eigenvalues
as a direct conclusion of lemma \ref{lem:infinitely_many_generic_eigenvalues}.
If $\metgraph$ is a tree graph then it was proved in \cite{PokPryObe_mz96,Schapotschnikow06}
(see also appendix A in \cite{Berkolaiko07}) that the nodal counts
of all generic eigenfunctions are $\fc=n-1$ and $\nc=n$. Otherwise,
if $\metgraph$ has $\beta>0$ cycles, assume by contradiction that
there are only finitely many generic eigenfunctions with $\fc\neq n-1$.
In particular this means that there is at least one generic eigenfunction
for which $\fc=n-1$ and thus $\sur=0$. We may therefore conclude
from theorem \ref{thm:infinitely_many_with_same_or_symmetric_surplus}
that there are infinitely many generic eigenfunctions whose surplus
is $\beta$ ($\fc=n-1+\beta$) and get a contradiction. 

In order to prove the statement about the nodal domain count we use
the following connection between nodal point count and nodal domain
count for metric graphs, which is proved for example in \cite{BanBerRazSmi_cmp12}
(see there, equation (1.11) together with lemma 3.2), 
\begin{equation}
\exists n_{0}\,\textrm{ s.t. \,}\forall n>n_{0}\,\quad\fc=\nc-1+\beta.\label{eq:nodal_domains_expressed_as_nodal_counts_high_in_the_spectrum}
\end{equation}
Now, assume by contradiction that there are only finitely many eigenfunctions
with $\nc<n$. Then there exists some $n_{1}>n_{0}$ such that for
all $n\in\nums\cap\left(n_{1},\infty\right),\,\nc=n$ and from (\ref{eq:nodal_domains_expressed_as_nodal_counts_high_in_the_spectrum})
we get that for all $n\in\nums\cap\left(n_{1},\infty\right),\,\fc=n-1+\beta.$
The surplus of all these eigenfunction is $\beta$. Applying theorem
\ref{thm:infinitely_many_with_same_or_symmetric_surplus} gives that
there are also infinitely many generic eigenfunctions with surplus
equals to zero, but the conclusion in the previous sentence allows
for only finitely many surpluses to differ than $\beta$ and hence
the contradiction.
\end{proof}
We now proceed to the proof of lemma \ref{lem:infinitely_many_generic_eigenvalues},
and bring two additional lemmata, all of which will be used to prove
theorem \ref{thm:infinitely_many_with_same_or_symmetric_surplus}.
In the following we refer to $k_{n}\left(\magv\right)=\sqrt{\lambda_{n}\left(\magv\right)}$
as the graph eigenvalues, which is valid and done for the sake of
simplicity.
\begin{proof}
[Proof of lemma \ref{lem:infinitely_many_generic_eigenvalues}]Let
$\vec{\xi}\in\R^{I}$ be a vector of incommensurate entries which
express the graph edge lengths by (\ref{eq:edge_lengths_as_linear_combinations}).
Let $k_{0}$ be a generic eigenvalue of $\metgraph$. Note that the
graph eigenvalues are given as zeros of $F$ together with their multiplicities
(a property that $F$ inherits from $\tilde{F}$, see (\ref{eq:secular_function})
and (\ref{eq:secular_equation_defined_on_torus})). We therefore get
that $\frac{\textrm{d}}{\textrm{d}k}F\neq0$ at simple eigenvalues.
As we assumed $k_{0}$ to be a simple eigenvalue we have $\left.\frac{\textrm{d}}{\textrm{d}k}F\right|_{\left(k_{0}\vec{\xi},\zvec\right)}=\left.\vec{\xi}\cdot\vec{\nabla}F\right|_{\left(k_{0}\vec{\xi},\zvec\right)}\neq0$.
We know that $k_{0}\vec{\xi}\in\Sigma_{\zvec}$, as $k_{0}$ is an
eigenvalue and $\Sigma_{\zvec}$ is the zero set of $F$ (recall section
\ref{sub:ergodic_motion_on_the_graph_torus}). Choose a neighbourhood
$\Xi$ of $k_{0}\vec{\xi}$ on $\Sigma_{\zvec}$ such that $\left.\vec{\xi}\cdot\vec{\nabla}F\right|_{\left(\tvec,\zvec\right)}\neq0$
for all $\tvec\in\Xi$. Recall (remark \ref{rem:eigenfunctions_on_the_torus})
that the values the eigenfunction obtains on the graph vertices are
a continuous function of $\tvec$. From the genericity of $k$ we
know that its corresponding eigenfunction does not vanish on the vertices
and we may therefore choose $\Xi$ small enough such that this property
holds for all $\tvec\in\Xi$. As the flow is ergodic, the set $\Xi\subset\Sigma_{\zvec}$
will be pierced an infinite number of times by the flow, yielding
infinitely many eigenvalues. Our choice of $\Xi$ guarantees both
that these eigenvalues are simple and that their corresponding eigenfunctions
do not vanish on vertices.\end{proof}
\begin{lem}
\label{lem:hessian_of_eigenvalue_equals_hessian_on_torus}Let $\vec{\xi}\in\R^{I}$
be a vector of incommensurate entries which express the graph edge
lengths by (\ref{eq:edge_lengths_as_linear_combinations}). Let $k\left(\zvec\right)\neq0$
be a generic eigenvalue. The Hessian of this eigenvalue with respect
to the magnetic fluxes at $\magv=\zvec$ is given as
\begin{equation}
\hes_{k}\left(\zvec\right)=-\left.\frac{\hes_{F}}{\left(\vec{\xi}\cdot\vec{\nabla}F\right)}\right|_{\left(k\left(\zvec\right)\vec{\xi};\zvec\right)},\label{eq:hessian_of_eigenvalue_equals_hessian_on_torus}
\end{equation}
where $\hes_{F}$ denotes the Hessian of the secular function $F\left(k\left(\zvec\right)\vec{\xi};\cdot\right)$
with respect to its magnetic parameters and $\vec{\nabla}F$ is the
gradient of $F\left(\cdot;\zvec\right)$ taken with respect to its
coordinates on the torus $\torus$.\end{lem}
\begin{proof}
The eigenvalue $k\left(\magv\right)$ is given implicitly as the solution
of $F\left(k\left(\magv\right)\vec{\xi};\,\magv\right)=0$. Take the
second total derivatives of this expression with respect to any two
magnetic fluxes, 
\begin{eqnarray*}
0 & = & \frac{\textrm{d}^{2}}{\textrm{d}\alpha_{i}\textrm{d}\alpha_{j}}F\left(k\left(\magv\right)\vec{\xi};\,\magv\right)\\
 & = & \frac{\partial^{2}k}{\partial\alpha_{i}\partial\alpha_{j}}\left(\vec{\xi}\cdot\vec{\nabla}F\right)+\frac{\partial k}{\partial\alpha_{i}}\frac{\partial k}{\partial\alpha_{j}}\sum_{p,q=1}^{\left|I\right|}\frac{\partial^{2}F}{\partial x_{p}\partial x_{q}}\xi_{p}\xi_{q}+\frac{\partial k}{\partial\alpha_{j}}\left(\vec{\xi}\cdot\frac{\partial}{\partial\alpha_{j}}\vec{\nabla}F\right)+\frac{\partial^{2}F}{\partial\alpha_{i}\partial\alpha_{j}}.
\end{eqnarray*}
Evaluate the above at the point $\magv=\zvec$ and use $\forall i\,\left.\frac{\partial k}{\partial\alpha_{i}}\right|_{\magv=\zvec}=0$
(theorem \ref{thm:magnetic_morse_index}, part (\ref{enu:magnetic_morse_index_part_1})),
to get that the second and the third terms above vanish and we are
left with
\begin{equation}
\left.\left\{ \left(\vec{\xi}\cdot\vec{\nabla}F\right)\hes_{k}\left(\zvec\right)+\hes_{F}\right\} \right|_{\left(k\left(\zvec\right)\vec{\xi},\zvec\right)}=0.\label{eq:hessian_of_eigenvalue_and_hessian_on_torus_relation}
\end{equation}
We now repeat the argument in the proof of lemma \ref{lem:infinitely_many_generic_eigenvalues}
to conclude that if $k\left(\zvec\right)$ is simple then $\left.\frac{\textrm{d}}{\textrm{d}k}F\right|_{\left(k\left(\zvec\right)\vec{\xi},\zvec\right)}=\left.\vec{\xi}\cdot\vec{\nabla}F\right|_{\left(k\left(\zvec\right)\vec{\xi},\zvec\right)}\neq0$.
Thus, we divide (\ref{eq:hessian_of_eigenvalue_and_hessian_on_torus_relation})
by $\left.\vec{\xi}\cdot\vec{\nabla}F\right|_{\left(k\left(\zvec\right)\vec{\xi},\zvec\right)}$
to get (\ref{eq:hessian_of_eigenvalue_equals_hessian_on_torus}).\end{proof}
\begin{lem}
\label{lem:symmetry_of_secular_function}The secular function $F$
of a Neumann graph exhibits the following symmetry
\begin{equation}
F\left(\tvec;\,\magv\right)=F\left(-\tvec;\,-\magv\right).\label{eq:symmetry_of_secular_function}
\end{equation}
\end{lem}
\begin{proof}
Let $\tvec\in\torus$ and $\alpha\in\R^{\beta}$. Choose some $\vec{\xi}\in\R^{\left|I\right|}$
with incommensurate entries. There exists $k\in\R$ such that $\tvec\equiv k\vec{\xi}$
modulo $\torus$. The required symmetry (\ref{eq:symmetry_of_secular_function})
is now obtained from the following series of equalities.
\begin{eqnarray*}
F\left(k\vec{\xi};\,\magv\right)= & \tilde{F}\left(k;\,\lenmap\left(\vec{\xi}\right);\,\magv\right) & =\overline{\tilde{F}\left(-k;\,\lenmap\left(\vec{\xi}\right);\,-\magv\right)}\\
 &  & =\tilde{F}\left(-k;\,\lenmap\left(\vec{\xi}\right);\,-\magv\right)=F\left(-k\vec{\xi};\,-\magv\right).
\end{eqnarray*}
The first and last equalities result from property (\ref{enu:secular_function_second_property}),
mentioned after the definition of $F$ on the torus, (\ref{eq:secular_equation_defined_on_torus}).
The second equality can be deduced from (\ref{eq:secular_function})
together with the linearity of $\mathbf{\magmat}\left(\magv\right)$
and the observation that $\scatmat\left(k\right)$ is real and $k$-independent
for a Neumann graph. Finally, the third inequality (from the first
to the second line) is due to $\tilde{F}$ being real.
\end{proof}
~
\begin{proof}
[Proof of theorem \ref{thm:infinitely_many_with_same_or_symmetric_surplus}]
Denote the edge lengths of $\metgraph$ by the vector $\lenv$ and
choose an incommensurate set $\vec{\xi}\in\R^{I}$ which is related
to graph edge lengths by (\ref{eq:edge_lengths_as_linear_combinations}).
Let $k$ be a generic eigenvalue of $\metgraph$ whose nodal surplus
is $\sigma$. The Hessian of $k$ is given by (\ref{eq:hessian_of_eigenvalue_equals_hessian_on_torus})
in lemma \ref{lem:hessian_of_eigenvalue_equals_hessian_on_torus}
and we know from theorem \ref{thm:magnetic_morse_index} part (\ref{enu:magnetic_morse_index_part_2})
that it is non-degenerate. We may therefore choose a neighbourhood
$\Xi$ of $k\vec{\xi}$ on $\Sigma_{\zvec}$ such that $\left.\frac{\hes_{F}}{\left(\vec{\xi}\cdot\vec{\nabla}F\right)}\right|_{\left(\tvec,\zvec\right)}$
is non-degenerate for all $\tvec\in\Xi$, and thus the number of negative
eigenvalues of $\left.\frac{\hes_{F}}{\left(\vec{\xi}\cdot\vec{\nabla}F\right)}\right|_{\left(\tvec,\zvec\right)}$
stays constant in this neighbourhood. This number equals to the Morse
index of the eigenvalue and it equals to $\sigma$ by an application
of theorem \ref{thm:magnetic_morse_index} (part (\ref{enu:magnetic_morse_index_part_3})).
As the flow is ergodic, the set $\Xi\subset\Sigma_{\zvec}$ is pierced
an infinite number of times by the flow, yielding infinitely many
eigenvalues. All of these eigenvalues are generic if $\Xi$ is chosen
small enough, as can be shown by repeating the argument in the proof
of lemma \ref{lem:infinitely_many_generic_eigenvalues}.\textbf{ }Each
of these eigenvalues would have this same Morse index, $\sigma$,
by (\ref{eq:hessian_of_eigenvalue_equals_hessian_on_torus}). The
nodal surplus of each of these eigenvalues is therefore also equal
to $\sigma$, according to theorem \ref{thm:magnetic_morse_index}
(part (\ref{enu:magnetic_morse_index_part_3})) and this proves the
first statement in our theorem. The second statement is proved with
the aid of the symmetry (\ref{eq:symmetry_of_secular_function}) in
lemma \ref{lem:symmetry_of_secular_function}, which allows to conclude
\begin{eqnarray*}
\forall\tvec\in\Xi\,\,\,\,;\left.\hes_{F}\right|_{\left(\tvec,\zvec\right)} & = & \left.\hes_{F}\right|_{\left(-\tvec,\zvec\right)}\\
\left.\vec{\nabla}F\right|_{\left(\tvec,\zvec\right)} & = & -\left.\vec{\nabla}F\right|_{\left(-\tvec,\zvec\right)},
\end{eqnarray*}
from which we get 
\begin{eqnarray}
\forall\tvec\in\Xi\,\,\,\,;-\left.\frac{\hes_{F}}{\left(\vec{\xi}\cdot\vec{\nabla}F\right)}\right|_{\left(\tvec,\zvec\right)} & = & \left.\frac{\hes_{F}}{\left(\vec{\xi}\cdot\vec{\nabla}F\right)}\right|_{\left(-\tvec,\zvec\right)}.\label{eq:inverse_hessians}
\end{eqnarray}
The LHS has $\sigma$ negative eigenvalues and therefore the RHS has
$\beta-\sigma$ negative eigenvalues. From (\ref{eq:symmetry_of_secular_function})
we also deduce $\Xi\subset\Sigma_{\zvec}\,\Rightarrow\,-\Xi\subset\Sigma_{\zvec}$
(where $-\Xi:=\left\{ \tvec;\,\,-\tvec\in\Xi\right\} $). Now use
again the ergodicity of the flow to conclude that the set $-\Xi$
is pierced an infinite number of times by the flow and all resulting
eigenvalues have Morse index $\beta-\sigma$, which equals to their
nodal surplus following theorem \ref{thm:magnetic_morse_index}.\end{proof}
\begin{rem}
\label{rem:spectral_invariant_for_metric_case}The proof above essentially
enables the proof of the metric inverse nodal theorem (theorem \ref{thm:main_theorem_metric}).
One might note that the proof here is of a different nature than the
proof of the discrete inverse nodal theorem (theorem \ref{thm:main_theorem_discrete}),
where we have identified the trace of the operator as a spectral invariant
independent of magnetic potential. One can find, however, a similarity
between the proofs, as the antisymmetric relation (\ref{eq:inverse_hessians})
points on the possibility to average the Hessians over whole torus
and show that the magnetic dependence cancels.
\begin{rem}
In the course of the proof we have used an equality between the Morse
index of a particular eigenvalue and the Morse index of the secular
function $F$ evaluated at the appropriate point. A similar relation
appears in appendix E of \cite{cdv_arx12} where the Morse index of
an eigenvalue is expressed in terms of the Morse index of the characteristic
polynomial evaluated at this eigenvalue.
\begin{rem}
The symmetry which lemma \ref{lem:symmetry_of_secular_function} describes
can already be exhibited on the level of the Schr�dinger operator
on the graph. One can show that the $\magv\rightleftarrows-\magv$
symmetry amounts to conjugation of all the eigenvalues (see (\ref{eq:magnetic_metric_Schroedinger})),
which actually leaves the spectrum invariant as it is real. Lemma
\ref{lem:symmetry_of_secular_function} tells us that the transformation
$\magv\rightarrow-\magv$ changes the sign of the $k$-eigenvalues,
which also leaves the spectrum invariant. This spectral symmetry was
exploited in \cite{BerWey_ptrs13,cdv_arx12} to prove that the eigenvalues
have critical points at $\magv=\zvec$.
\end{rem}
\end{rem}
\end{rem}

\section{Discretized Versions of a Metric Graph\label{sec:discretized_vesrions}}

This section presents a connection between a certain metric graph
and various discrete graphs with similar nodal surplus. These discrete
graphs can be obtained from the metric one by means of a simple construction
and they will be called its discretized versions. This correspondence
between a metric graph and its discretized versions is interesting
on its own and can also be used as a tool to provide an additional
proof for theorem \ref{thm:infinitely_many_with_same_or_symmetric_surplus}
and thus also for the inverse result in theorem \ref{thm:main_theorem_metric}.
More interestingly, this might be used to further extend theorem \ref{thm:main_theorem_metric}
for graphs with general vertex conditions and electric potentials
(see remark \ref{rem:generalizations_for_non_neumann_and_non_zero_potential}).

The construction starts by picking a vector $\vec{\xi}\in\R^{I}$
of incommensurate entries such that the graph edge lengths are given
by $\lenv=\lenmap\left(\vec{\xi}\right)$ (as in section \ref{sub:ergodic_motion_on_the_graph_torus}).
The specific discretized version we construct is characterized by
a selection of 
\begin{equation}
\intlenv\in\textrm{Image}\left(\lenmap\right)\cap\N^{\EE}.\label{eq:integer_lengths_of_metric_graph}
\end{equation}
Note that the set $\textrm{Image}\left(\lenmap\right)\cap\N^{\EE}$
is non-empty. In order to show this, one can approximate $\vec{\xi}$
by some rational vector, $\vec{\xi}^{rat}\in\Q^{I}$, such that $\lenmap\left(\vec{\xi}^{rat}\right)$
has all positive entries (just as $\lenv=\lenmap\left(\vec{\xi}\right)$
does). Note that $\lenmap\left(\vec{\xi}^{rat}\right)\in\Q^{\EE}$
(see (\ref{eq:edge_lengths_as_linear_combinations})), and therefore
the vector $\lenmap\left(\vec{\xi}^{rat}\right)$ can now be multiplied
by a common divisor of its entries to turn its entries to natural
numbers, while retaining it in $\textrm{Image}\left(\lenmap\right)$.

Take the underlying discrete graph of the metric graph, $\metgraph$,
and equip each edge $e\in\EE$ with $j_{e}-1$ new vertices of degree
two, which will split this edge into $j_{e}$ new edges. This (new)
discrete graph is denoted $\dismetgraph$ and called \emph{a discretized
version of} $\metgraph$. Note that a discretized version is not uniquely
determined by $\metgraph$. The set of all possible discretized versions
is given by $\textrm{Image}\left(\lenmap\right)\cap\N^{\EE}$. This
set depends on the 'nature of incommensurability' of the original
edge lengths, $\left\{ l_{e}\right\} _{e\in\EE}$, i.e., the rational
dependencies between these lengths. However, one can verify that $\textrm{Image}\left(\lenmap\right)\cap\N^{\EE}$
does not depend on the particular choice of incommensurate representatives,
$\left\{ \xi_{i}\right\} _{i\in I}$.

Note that one may convert the discretized graph, $\dismetgraph$,
back into a metric graph (different than $\metgraph$) by setting
all of $\dismetgraph$'s edge lengths to equal one. One would then
get a metric graph with the same connectivity as $\metgraph$, but
with integer edge lengths given by $\intlenv$. We denote this metric
graph by $\tilde{\Gamma}$ and it will turn to be useful in the course
of proving the following.
\begin{thm}
\label{thm:discrete_metric_surplus_connection}Let $\metgraph$ be
a Neumann metric graph with $\beta>0$ cycles\emph{ }and let $\dismetgraph$
be a discretized version of $\metgraph$. Let $\mu\notin\left\{ 0,2\right\} $
be a generic eigenvalue of $\lap^{\left(norm\right)}$ on $\dismetgraph$
whose nodal surplus is $\sigma$. Then
\begin{enumerate}
\item $\metgraph$ has infinitely many generic eigenvalues whose nodal surplus
equals $\sigma$.
\item $\metgraph$ has infinitely many generic eigenvalues whose nodal surplus
equals $\beta-\sigma$.
\end{enumerate}
\end{thm}
\begin{rem}
It is important to emphasize that the theorem holds for \textbf{all}
discretized versions, $\dismetgraph$, of a metric graph, $\metgraph$.
\begin{rem}
\label{rem:discrete_eigenvalues_are_0_2}One should note that the
theorem is empty if the spectrum of the chosen discretized version
consists only of the eigenvalues $\left\{ 0,2\right\} $ and other
non-generic eigenvalues. The only connected graph which has no eigenvalues
different from $\left\{ 0,2\right\} $ is the single edge graph (see
for example, lemma 1.8 in \cite{Chung_spectralgraph}). Nevertheless
it does not fit our theorem as it has no cycles. Yet, there are discrete
graphs (with cycles) whose all eigenvalues which are different than
$\left\{ 0,2\right\} $ are not simple. The $n-$cube graph forms
such an example (example 1.6 in \cite{Chung_spectralgraph}). Following
the previous remark, one may still wonder whether for given a metric
graph, there is always \textbf{some} discretized version for which
theorem \ref{thm:discrete_metric_surplus_connection} is not empty.
We are not aware of possible answers to this question.
\end{rem}
\end{rem}
\begin{proof}
Start by proving the theorem for an equilateral graph $\metgraph$,
and a specific discretization of it. We may assume without loss of
generality that for all $e\in\EE$, $l_{e}=1$ (as nodal count is
indifferent to scaling). We choose a discretized version, $\dismetgraph$,
which has the same vertex and edge sets and connectivity as $\metgraph$
does (it is obtained by choosing $\intlenv=\left(1,\ldots,1\right)$
in (\ref{eq:integer_lengths_of_metric_graph})). Let $\mu\left(\zvec\right)\notin\left\{ 0,2\right\} $
be a generic eigenvalue of $\lap^{\left(norm\right)}\left(\zvec\right)$
on $\dismetgraph$ whose eigenvector is $f$ and nodal surplus is
$\sigma$. Consider the magnetic Laplacian, $\lap^{\left(norm\right)}\left(\magv\right)$,
apply theorem \ref{thm:magnetic_morse_index} and obtain that $\mu\left(\magv\right)$
has a critical point at $\magv=\zvec$ and its Morse index is $\morse_{\mu}\left(\zvec\right)=\sigma$.
Theorem \ref{thm:discrete_metric_spectral_connection} shows that
$\left\{ \left(\arccos\left[1-\mu\left(\magv\right)\right]\right)^{2}\right\} $
are eigenvalues of the magnetic Schr�dinger operator on $\metgraph$
and we wish to obtain their Morse indices at $\magv=\zvec$. First,
exclude the values $\pm1$ from the domain of $\arccos$ and consider
it as a union of single valued 'branch' functions $\left\{ b_{p}:\left(-1,1\right)\rightarrow\left(p\pi,\left(p+1\right)\pi\right))\right\} _{p=0}^{\infty}$.
With this notation, those eigenvalues of $\metgraph$ resulting from
$\mu\left(\magv\right)$ are given by 
\[
\lambda_{\left(p\right)}\left(\magv\right)=\left(b_{p}\left[1-\mu\left(\magv\right)\right]\right)^{2},
\]
where the subscript $_{\left(p\right)}$ is not to be confused with
the serial number of $\lambda_{\left(p\right)}$ in $\metgraph$'s
spectrum. Each of those eigenvalues is obtained as a function of $\mu\left(\magv\right)$,
which has a well defined monotonicity:
\begin{enumerate}
\item For even values of $p$, $\lambda_{\left(p\right)}\left(\mu\left(\magv\right)\right)$
is a monotone strictly increasing function of $\mu\left(\magv\right)$.
The Hessians of $\lambda_{\left(p\right)}\left(\magv\right)$ and
$\mu\left(\magv\right)$ at $\magv=\zvec$ are therefore equal up
to a positive multiple, which yields equality of their Morse indices,
\begin{equation}
\morse_{\lambda_{\left(p\right)}}\left(\zvec\right)=\morse_{\mu}\left(\zvec\right).\label{eq:morse_indices_relation_1}
\end{equation}

\item For odd values of $p$, $\lambda_{\left(p\right)}\left(\mu\left(\magv\right)\right)$
is a monotone strictly decreasing function of $\mu\left(\magv\right)$.
Therefore, in this case the Hessians of $\lambda_{\left(p\right)}\left(\magv\right)$
and $\mu\left(\magv\right)$ at $\magv=\zvec$ are equal up to a negative
multiple, which yields the following relation of their Morse indices
\begin{equation}
\morse_{\lambda_{\left(p\right)}}\left(\zvec\right)=\beta-\morse_{\mu}\left(\zvec\right).\label{eq:morse_indices_relation_2}
\end{equation}

\end{enumerate}
We may now apply the metric version of theorem \ref{thm:magnetic_morse_index}
for the eigenvalues $\lambda_{\left(p\right)}$, but first verify
that they satisfy the theorem's assumptions. Their simplicity (at
$\magv=\zvec$) follows from theorem \ref{thm:discrete_metric_spectral_connection}
as $\mu\left(\zvec\right)$ is simple and $\mu\left(\zvec\right)\notin\left\{ 0,2\right\} $
(which guarantees that $\lambda_{\left(p\right)}\left(\zvec\right)\notin\left\{ \left(n\pi\right)^{2}\right\} _{n\in\Z}$,
i.e. different from the Dirichlet eigenvalues). In addition, according
to theorem \ref{thm:discrete_metric_spectral_connection} (part (\ref{enu:discrete_metric_thm_eigenfunction_on_vertices})),
the restriction of the eigenfunction of $\lambda_{\left(p\right)}\left(\zvec\right)$
to the graph vertices equals $\degmat^{\nicefrac{1}{2}}f$ where $f$,
the eigenvector corresponding to $\mu\left(\zvec\right)$, is different
than zero on all vertices, by the theorem's assumption. The nodal
surplus of $\lambda_{\left(p\right)}$, which we denote by $\sigma^{\left(\metgraph\right)}\left(\lambda_{\left(p\right)}\right)$
can therefore be expressed as 
\begin{eqnarray}
\textrm{for even}\, p,\,\,\,\sigma^{\left(\metgraph\right)}\left(\lambda_{\left(p\right)}\right) & = & \morse_{\lambda_{\left(p\right)}}\left(\zvec\right)=\morse_{\mu\left(\magv\right)}\left(\zvec\right)=\sigma\label{eq:equality_morse_indices_metric_discrete}\\
\textrm{for odd}\, p,\,\,\,\,\sigma^{\left(\metgraph\right)}\left(\lambda_{\left(p\right)}\right) & = & \morse_{\lambda_{\left(p\right)}}\left(\zvec\right)=\beta-\morse_{\mu\left(\magv\right)}\left(\zvec\right)=\beta-\sigma,\label{eq:equality_morse_indices_metric_discrete_reversed}
\end{eqnarray}
which proves the theorem for the case of an equilateral metric graph
if  its discretization given by choosing $\intlenv=\left(1,\ldots,1\right)$.

If the graph $\metgraph$ is not equilateral and the discretized version,
$\dismetgraph$, is arbitrary there is no exact expression which connects
both spectra of $\metgraph$ and $\dismetgraph$. The route we take
this time is to turn $\dismetgraph$ into an equilateral graph, $\tilde{\metgraph}$,
all of whose edge lengths equal one. Therefore, there are Morse index
connections similar to (\ref{eq:morse_indices_relation_1}), (\ref{eq:morse_indices_relation_2})
between $\tilde{\Gamma}$ and $\dismetgraph$. We then show that infinitely
many eigenvalues of $\metgraph$ share the same Morse index (and hence
the same surplus) as eigenvalues of $\tilde{\Gamma}$ and this yields
the desired statements in the theorem. This is the content of the
rest of the proof.

Let $\mu\left(\zvec\right)\notin\left\{ 0,2\right\} $ be a generic
eigenvalue of $\lap^{\left(norm\right)}\left(\zvec\right)$ on $\dismetgraph$.
We may therefore conclude, just as in the first part of the proof,
that $\hes_{\mu}\left(\zvec\right)$ equals up to a factor the Hessians
of infinitely many eigenvalues of $\tilde{\metgraph}$ (half of these
factors are positive and half are negative). Denote by $\intlenv$
the vector which generates the discretization $\dismetgraph$. Consider
$\tilde{\metgraph}$ as a metric graph with the same connectivity
as $\metgraph$, but with $\left\{ \intlen_{e}\right\} _{e\in\EE}$
as the set of its edge lengths. The $k$-eigenvalues of $\tilde{\Gamma}$
are then described by 
\[
\left\{ k\left(\magv\right);\,\, k\intlenv\in\Sigma_{\magv}\right\} ,
\]
where $\Sigma_{\magv}=\left\{ \tvec;\,\, F\left(\tvec;\,\magv\right)=0\right\} $.
Namely, the same torus, $\torus$, can be used to describe the eigenvalues
of $\metgraph$ and $\tilde{\metgraph}$, but with different flow
directions ($\vec{\xi}$ for $\metgraph$ and $\intlenv$ for $\tilde{\metgraph}$).
As for the spectral connection to $\dismetgraph$, we know from theorem
\ref{thm:discrete_metric_spectral_connection} that $\tilde{\metgraph}$
has the $k$-eigenvalues $\left\{ b_{p}\left[1-\mu\left(\magv\right)\right]\right\} _{p\in\N\cup\left\{ 0\right\} }$.
Choose two $k$-eigenvalues with different parity of $p$, for example
\begin{eqnarray*}
k_{0}\left(\magv\right) & := & b_{0}\left[1-\mu\left(\magv\right)\right]\\
k_{1}\left(\magv\right) & := & b_{1}\left[1-\mu\left(\magv\right)\right].
\end{eqnarray*}
From a similar monotonicity argument, as in the first part of the
proof, we get that 
\begin{equation}
\hes_{k_{0}}\left(\zvec\right)=-c\hes_{k_{1}}\left(\zvec\right),\label{eq:hessians_k0_k1}
\end{equation}
where $c>0$ and these Hessians are non-degenerate. Use relation (\ref{eq:hessian_of_eigenvalue_equals_hessian_on_torus})
in lemma \ref{lem:hessian_of_eigenvalue_equals_hessian_on_torus}
to write
\[
\hes_{k_{i}}\left(\zvec\right)=-\left.\frac{\hes_{F}}{\left(\intlenv\cdot\vec{\nabla}F\right)}\right|_{\left(k_{i}\left(\zvec\right)\intlenv;\zvec\right)}\,\textrm{for}\,\, i=0,1.
\]
As these Hessians are non-degenerate, we may choose neighbourhoods
$\Xi_{i}$ $\left(i=0,1\right)$ of $k_{i}\left(\zvec\right)\cdot\intlenv$
on $\Sigma_{\magv}$ such that $\forall\tvec\in\Xi_{i},\,\left.\frac{\hes_{F}}{\left(\intlenv\cdot\vec{\nabla}F\right)}\right|_{\left(\tvec;\zvec\right)}$is
non-degenerate and define the corresponding Morse index function
\begin{eqnarray}
m_{\intlenv}:\left(\Xi_{0}\cup\Xi_{1}\right) & \rightarrow & \N\cup\left\{ 0\right\} \nonumber \\
m_{\intlenv}\left(\tvec\right) & := & \morse\left(\left.\frac{\hes_{F}}{\left(\intlenv\cdot\vec{\nabla}F\right)}\right|_{\left(\tvec;\zvec\right)}\right),\label{eq:Morse_index_function_definition}
\end{eqnarray}
where $\morse\left(\cdot\right)$ returns the number of negative eigenvalues.
Note that $m_{\intlenv}$ is constant on each set $\Xi_{i}$ , due
to non-degeneracy of the Hessians there. From (\ref{eq:hessians_k0_k1})
we get the following relation on these Hessians
\[
\forall\tvec_{0}\in\Xi_{0},\tvec_{1}\in\Xi_{1}\,\,\,\, m_{\intlenv}\left(\tvec_{0}\right)+m_{\intlenv}\left(\tvec_{1}\right)=\beta.
\]

Let us now return to the original graph, $\metgraph$. The flow which
characterizes its eigenvalues goes in the direction $\vec{\xi}$ within
the torus $\torus$ and we wish to adapt the definition of the Morse
index function, (\ref{eq:Morse_index_function_definition}), to this
flow. We now show that $\left.\intlenv\cdot\vec{\nabla}F\right|_{\left(\tvec;\zvec\right)}$
and $\left.\vec{\xi}\cdot\vec{\nabla}F\right|_{\left(\tvec;\zvec\right)}$
have the same sign on $\Xi_{i}$ and conclude that 
\[
\left.m_{\vec{\xi}}\right|_{\Xi_{0}\cup\Xi_{1}}\equiv\left.m_{\intlenv}\right|_{\Xi_{0}\cup\Xi_{1}}.
\]
As the flow defined by $\vec{\xi}$ pierces both $\Xi_{0}$ and $\Xi_{1}$
an infinite number of times, we get an infinite number of eigenvalues
of $\metgraph$ whose Morse indices are $\sigma$ and $\beta-\sigma$.
Applying theorem \ref{thm:magnetic_morse_index} finishes the proof.

All remains is therefore to show that $\left.\intlenv\cdot\vec{\nabla}F\right|_{\left(\tvec;\zvec\right)}$
and $\left.\vec{\xi}\cdot\vec{\nabla}F\right|_{\left(\tvec;\zvec\right)}$
have the same sign (for any $\tvec\in\torus$ which represents a simple
eigenvalue). We have shown in the course of the proof of lemma \ref{lem:infinitely_many_generic_eigenvalues}
that $\left.\intlenv\cdot\vec{\nabla}F\right|_{\left(\tvec;\zvec\right)}$
cannot vanish as $\tvec$ represents a simple eigenvalue. We claim
that this expression does not vanish even if we replace $\vec{\intlen}$
with any convex linear combination of $\vec{\intlen}$ and $\vec{\xi}$.
Namely, we show that for all $t\in\left[0,1\right],\,\,\left.\left(t\intlenv+\left(1-t\right)\vec{\xi}\right)\cdot\vec{\nabla}F\right|_{\left(\tvec;\zvec\right)}\neq0$
from which it is clear that $\left.\intlenv\cdot\vec{\nabla}F\right|_{\left(\tvec;\zvec\right)}$
and $\left.\vec{\xi}\cdot\vec{\nabla}F\right|_{\left(\tvec;\zvec\right)}$
have the same sign. Finally, the observation $\forall t\in\left[0,1\right],\,\,\left.\left(t\intlenv+\left(1-t\right)\vec{\xi}\right)\cdot\vec{\nabla}F\right|_{\left(\tvec;\zvec\right)}\neq0$
is explained using the simplicity of the eigenvalue. The simplicity
of the eigenvalue guarantees that $\Sigma_{\magv}$ is homotopic to
an $\left|I\right|-1$ dimensional disc in a close neighbourhood of
$\tvec$. For each $t\in\left[0,1\right]$, the flow in the direction
$t\intlenv+\left(1-t\right)\vec{\xi}$ describes the eigenvalues of
a certain metric graph. This is since $\lenmap\left(t\intlenv+\left(1-t\right)\vec{\xi}\right)$
is a vector of positive entries which is true as $\lenmap\left(\intlenv\right)$
and $\lenmap\left(\vec{\xi}\right)$ have this property and the length
map, $\lenmap\left(t\intlenv+\left(1-t\right)\vec{\xi}\right)$ is
linear (see (\ref{eq:length_map})). At some point, the flow in the
direction $t\intlenv+\left(1-t\right)\vec{\xi}$ will pierce $\Sigma_{\magv}$
at $\tvec$ and generate an eigenvalue. If we had $\left.\left(t\intlenv+\left(1-t\right)\vec{\xi}\right)\cdot\vec{\nabla}F\right|_{\left(\tvec;\zvec\right)}=0$,
this would mean that the flow is tangential to $\Sigma_{\magv}$ (with
$\tvec$ being the touching point). This would allow for a slight
perturbation of the direction of the flow in a way which would remove
the relevant eigenvalue from the spectrum (or alternatively, would
create a new one in its vicinity), which is not possible.\end{proof}
\begin{rem}
\label{rem:generalizations_for_non_neumann_and_non_zero_potential}This
theorem can be compared with theorem \ref{thm:infinitely_many_with_same_or_symmetric_surplus},
as their conclusions are similar. Yet, the current theorem seems weaker
as it requires more conditions (see remark \ref{rem:discrete_eigenvalues_are_0_2}
for example). The proof, however, contains an element which allows
to bypass the need of symmetry (lemma \ref{lem:symmetry_of_secular_function})
which is required in the proof of theorem \ref{thm:infinitely_many_with_same_or_symmetric_surplus}
and thus gives the possibility to generalize this result to non-Neumann
graphs and graphs with potentials. It was shown in \cite{Pan_lmp06}
that theorem \ref{thm:discrete_metric_spectral_connection} holds
also for $\delta$-type conditions and for some electric potentials,
if the the inverse Hill discriminant is used instead of the $\arccos$.
Therefore, one might try to repeat the proof above, replacing the
$\arccos$ with the corresponding inverse Hill discriminant. The mechanism
of ergodic flow on the torus does not describe accurately the eigenvalues
of $\delta$-type conditions, but it does so asymptotically, which
should be enough for our purpose. The case of electric potential might
be treated as well, as it was shown in recent work \cite{RueSmi_jpa12}
that its spectrum can be described asymptotically by secular function
similar to (\ref{eq:secular_function}).
\end{rem}

\section{A discussion\label{sec:discussion}}

The main result of this paper, as is implied by its title, is the
solution of the inverse nodal problem of determining a tree graph.
The solution is similar for both metric and discrete graphs - the
nodal point count sequence $\left\{ 0,1,2,3,\ldots\right\} $ may
be possessed solely by tree graphs. The similarity between discrete
and metric graphs carries over to the proofs - both use as a crucial
tool the recently established connection between the graph's nodal
count and dependence of its eigenvalues on magnetic fields, \cite{Ber_arx11,BerWey_ptrs13,cdv_arx12}.
The proof for discrete graphs is based on a very basic observation
- the trace of the operator does not depend on magnetic fields. The
proof for the metric case is of more exploratory type and concerns
properties of some individual eigenvalues. This proof yields some
additional results and offers further investigative directions. Yet,
it might seem superfluous for our main purpose. It would be interesting
to develop an alternative proof for the metric inverse problem which
tackles a specific spectral invariant similarly to the discrete case.
Possible candidates which arise as natural generalizations of the
trace are the vacuum energy with some regularization, or the value
of the zeta function at some point. Having said that, one should note
that the proof of theorem \ref{thm:infinitely_many_with_same_or_symmetric_surplus}
does include an implicit spectral invariant (see remark \ref{rem:spectral_invariant_for_metric_case}).
This work also sets some restrictions on possible nodal count sequences
which one can obtain from a graph. For example, we show that non-tree
graphs cannot have a nodal count sequence which is almost like the
tree nodal count (up to a finite number of discrepancies). In this
sense, our result resembles a recent 'quasi-isospectrality' result
by Rueckriemen, \cite{Rueckriemen_arx12}. He shows that if the spectra
of two graphs agree everywhere aside from a sufficiently sparse set,
then they are isospectral. In both his and our case, there are typical
sequences (either nodal or spectral) which characterize the graphs
and do not allow for 'small' number of discrepancies.

Putting aside the connection to the nodal count, one could phrase
the results we obtained as purely magnetic properties of graphs' spectra;
It is not possible for all graph eigenvalues to show diamagnetic behaviour
(see chapter 2 in \cite{FouHel_superconductivity_book}). This statement
holds for discrete graphs, whereas metric graphs obey even a stronger
restriction - an infinite number of eigenvalues must violate the diamagnetic
behaviour.

It is interesting to examine more deeply the genericity of the theorems
\ref{thm:main_theorem_discrete} and \ref{thm:main_theorem_metric}.
We have already claimed that assumption \ref{ass:generality_assumption}
is generic under various perturbations of the operator. Theorem \ref{thm:main_theorem_discrete}
was indeed proved with quite a high generality by allowing a dense
set of discrete Schr�dinger operators. Yet, one may wish to restrict
oneself to some particular classes of Laplacians (for example, the
adjacency matrix, $\conmat=\degmat-\degmat^{\nicefrac{1}{2}}\lap\degmat^{\nicefrac{1}{2}}$)
and ask for which graphs does theorem \ref{thm:main_theorem_discrete}
hold. This seems a rather involved question, which concerns a detailed
study of the graph automorphism group (see for example chapter 2 in
\cite{CveRowSim_eigenspaces97} and the references within). In the
case of metric graphs the assumptions in theorems \ref{thm:main_theorem_metric}
and \ref{thm:infinitely_many_with_same_or_symmetric_surplus} seem
milder than in the discrete case as only a single generic eigenvalue
is needed. However, we restrict ourselves to Neumann vertex conditions
and with no electric potentials. It is therefore desirable to generalize
the current results to include other vertex conditions and potentials.
A first step for doing so is suggested by the method of discretized
versions of a graph presented in section \ref{sec:discretized_vesrions}.
See also remark \ref{rem:generalizations_for_non_neumann_and_non_zero_potential}
which offers a possible approach for such generalizations. Of particular
interest is the question\textbf{ }whether for a given metric graph,
there is always some discretized version for which theorem \ref{thm:discrete_metric_surplus_connection}
is not empty. If his is the case theorem \ref{thm:discrete_metric_surplus_connection}
might serve as an alternative to theorem \ref{thm:infinitely_many_with_same_or_symmetric_surplus}
in the course of proof of the metric inverse result (theorem \ref{thm:main_theorem_metric}).

The inverse \emph{nodal domain }count problem for metric graphs was
solved as well in this paper, i.e., the nodal domain count sequence,
$\left\{ 1,2,3,\ldots\right\} $ implies the graph is a tree. Numerical
evidence suggests that this should be the case for discrete graphs
as well and it is interesting to prove (or maybe disprove) such an
inverse result. A possible approach might be to check if a similar
magnetic-nodal connection exist for the nodal domain count as well
(which would form a non-trivial generalization of the works \cite{Ber_arx11,BerWey_ptrs13,cdv_arx12})
and to apply it for the inverse problem. Another generalization direction
of the magnetic-nodal link, from which inverse problems would benefit,
is to give some treatment for non-simple eigenvalues and for eigenfunctions
which vanish at vertices. This is highly relevant, as the frequently
used discrete standard Laplacian tends to have such non-generic spectra.

Let us consider the inverse results in the paper from a spectral geometric
viewpoint. We have managed to distinguish a graph with cycles from
a tree by means of  the nodal count. The next almost immediate problem
would be to try and deduce the exact number of cycles out of the nodal
count, if this is possible at all. We know that under generic assumptions,
the number of cycles, $\beta$, is the upper bound of the nodal surplus.
One would then inquire whether this bound is attained somewhere in
the spectrum, i.e. whether 
\begin{equation}
\beta=\max_{n}\left\{ \fc-\left(n-1\right)\right\} ?\label{eq:beta_is_max_surplus?}
\end{equation}
In the discrete graphs setting, this is not the general case as shows
the following simple counter example

\begin{figure}[h]
\includegraphics[scale=0.6]{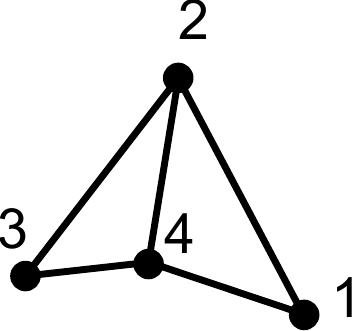}\hspace{10mm}$L=\begin{pmatrix}1 & -1 & 0 & -1\\
-1 & 2 & -1 & -1\\
0 & -1 & 3 & -1\\
-1 & -1 & -1 & 4
\end{pmatrix}$\caption{A simple graph with a certain choice for its Laplacian}
\label{fig:counter_example_1}
\end{figure}
The graph in figure \ref{fig:counter_example_1}(a) has $\beta=2$,
yet its  nodal point count is $\left\{ 0,2,3,3\right\} $ (with respect
to the given Laplacian matrix) and therefore $\max_{n}\left\{ \fc-\left(n-1\right)\right\} =1$.
Note however, that the nodal count depends not only on the graph,
but also on the chosen Laplacian for it (within the possible models
described in section \ref{sec:Introduction}). For example, had we
replaced the diagonal of the Laplacian with $\{4,3,2,1\}$ we would
have gotten the nodal count $\left\{ 0,3,3,4\right\} $, which does
attain the maximum in (\ref{eq:beta_is_max_surplus?}). We may therefore
still wonder whether (\ref{eq:beta_is_max_surplus?}) holds within
a certain class of Laplacians. If we consider the normalized Laplacian
on the same graph, for example, we get that the eigenvalues are simple,
but some eigenfunctions vanish at graph vertices and anyway (\ref{eq:beta_is_max_surplus?})
fails for the generic eigenvalues. The same is true for the so-called
standard Laplacian as well (same Laplacian as above, but with the
vertex degrees on its diagonal), where even not all eigenvalues are
simple. This calls for appropriate definitions of the  nodal count
for non-generic eigenvalues and possible generalizations of theorem
\ref{thm:magnetic_morse_index}. The question of deducing $\beta$
from the nodal count can be also asked for the metric graph setting.
It is possible to show that indeed under a genericity assumption one
may deduce the number of graph cycles out of the  nodal count. This
is to be discussed in a future work, where the nodal surplus distribution
of a metric graph is studied, \cite{BanBer_surplus_distribution}.
One should compare this inverse nodal problem to its analogous spectral
one - whether the spectrum reveals the number of graph cycles. This
is indeed the case, both for discrete graphs (see e.g., lemma 4 in
\cite{vanDamHae_laa03}) and for metric graphs (\cite{Kurasov_arkmat08}).
This comparison between the spectral data and the nodal one goes hand
to hand with the aforementioned conjecture in \cite{GSS05} that nodal
information resolves spectral ambiguity whenever it occurs. In the
discrete graph example shown in figure \ref{fig:counter_example_1},
however, the information flow is in the opposite way than the one
of the conjecture - the spectrum stores some geometric information
which the nodal count does not reveal. Yet, for metric graphs, as
was mentioned here, both spectral and nodal sequences tell the number
of graph cycles. Finally, we wish to mention trace formulae, which
are a major tool for treating inverse problems, as they connect spectral
information to geometry. For graphs, these formulae express spectral
functions in terms of closed paths on the graph. A trace formula for
the spectral counting function of a metric graph, \cite{KS97}, was
used to solve the inverse spectral problem on metric graphs \cite{GutSmi_jpa01}.
It was suggested by Smilansky that similar trace formulae exist for
functions of the nodal count and there is indeed some supporting evidence
and derivations for specific classes of two dimensional surfaces in
\cite{AroBanFajGnu_jpa12,AroSmi_arx10,GKS06}, and progress on the
problem for metric graphs in \cite{BanBerSmi_ahp12}. Yet, an exact
trace formula for the nodal count has not been found yet. Such a formula
will crucially advance solutions of inverse nodal problems of the
kind presented in this paper and of many more.

\section*{Acknowledgments}

The author wishes to acknowledge Lennie Friedlander who proposed him
the main inverse problem discussed in this paper. The author cordially
thanks Uzy Smilansky for numerous discussions of inverse nodal problems
and for his continuous encouragement. The ingenious insights offered
by Gregory Berkolaiko and Konstantin Pankrashkin are highly appreciated.
The author greatly benefited from discussions with Jon Keating, Michael
Levitin, Idan Oren and Adam Sawicki. The comments and corrections
offered by Gregory Berkolaiko and Tracy Weyand on the final version
helped to substantially improve the paper. This work is supported
by EPSRC, grant number EP/H028803/1.

\bibliographystyle{plain}
\bibliography{Ramis_Papers,Discrete_Graphs-Intro_Determinants_and_Magnetic,Miscellaneous,Nodal_Domains_Discrete_Graphs,Nodal_Domains_Metric_Graphs,Nodal_Domains_General,Qunatum_Graphs_Intro}

\end{document}